\documentclass[times,11pt,journal,onecolumn]{IEEEtran}

\usepackage{cite}
\usepackage{fixltx2e}
\usepackage[cmex10]{amsmath}
\usepackage{array}
\usepackage{wasysym}
\usepackage{dsfont}
\usepackage[latin1]{inputenc}
\usepackage[english]{babel}
\usepackage[T1]{fontenc}
\usepackage{mathtools}
\usepackage{amssymb}
\usepackage{enumerate}
\usepackage{bbm}
\usepackage{epsfig,syntonly}
\usepackage{verbatim,times}
\usepackage{epstopdf}
\usepackage{graphicx}
\usepackage{xcolor}
\usepackage{accents}
\usepackage{latexsym,fancyhdr,bm}
\usepackage[tight,footnotesize]{subfigure}
\usepackage{url}

\makeatletter
\let\l@ENGLISH\l@english
\makeatother

\newcommand{\fset}{{\mathcal F}}
\newcommand{\rset}{{\mathcal R}}
\newcommand{\iset}{{\mathcal I}}

\newcommand{\bset}{{\mathcal B}}
\newcommand{\dset}{{\mathcal D}}

\newcommand{\hset}{{\mathcal H}}
\newcommand{\lset}{{\mathcal L}}
\newcommand{\GAG}{AGG\;}

\newtheorem{theorem}{Theorem}
\newtheorem{proposition}{Proposition}

\newcommand{\btheo}{\begin{theorem}}
\newcommand{\etheo}{\end{theorem}}
\newcommand{\bproof}{\begin{proof}}
\newcommand{\eproof}{\end{proof}}
\newtheorem{definition}[theorem]{Definition}
\newcommand{\bdefi}{\begin{definition}}
\newcommand{\edefi}{\end{definition}}
\newtheorem{fact}[theorem]{Fact}
\newcommand{\bprop}{\begin{fact}}
\newcommand{\eprop}{\end{fact}}
\newtheorem{corollary}[theorem]{Corollary}
\newcommand{\bcor}{\begin{corollary}}
\newcommand{\ecor}{\end{corollary}}
\newtheorem{example}[theorem]{Example}
\newcommand{\bex}{\begin{example}}
\newcommand{\eex}{\end{example}}
\newtheorem{lemma}[theorem]{Lemma}
\newcommand{\blemma}{\begin{lemma}}
\newcommand{\elemma}{\end{lemma}}
\newtheorem{remark}[theorem]{Remark}
\newcommand{\bremark}{\begin{remark}}
\newcommand{\eremark}{\end{remark}}
\newtheorem{conj}[theorem]{Conjecture}
\newcommand{\bconj}{\begin{conj}}
\newcommand{\econj}{\end{conj}}

\def\0{{\tt 0}} 
\def\1{{\tt 1}} 
\def\?{{\tt *}} 
\renewcommand{\mid}{\,|\,}

\allowdisplaybreaks
\begin{document}

\title{Achieving Marton's Region for Broadcast Channels Using Polar Codes}
\author{Marco~Mondelli, S.~Hamed~Hassani, Igal~Sason, and~R\"{u}diger~Urbanke%
\thanks{M. Mondelli and R. Urbanke are with the School of
Computer and Communication Sciences, EPFL, CH-1015 Lausanne, Switzerland
(e-mails: \{marco.mondelli, ruediger.urbanke\}@epfl.ch).

S. H. Hassani is with the Computer Science Department, ETH Z\"{u}rich, Switzerland
(e-mail: hamed@inf.ethz.ch).

I. Sason is  with the Department of Electrical Engineering, Technion--Israel
Institute of Technology, Haifa 32000, Israel (e-mail: sason@ee.technion.ac.il).

The paper was presented in part at the {\em 48th Annual Conference on Information
Sciences and Systems (CISS~2014)}, Princeton, New Jersey, USA, March 2014, and
at the {\em 2014 IEEE International Symposium on Information Theory
(ISIT~2014)}, Honolulu, Hawaii, USA, July 2014.}}

\maketitle

\begin{abstract}
This paper presents polar coding schemes for the $2$-user
discrete memoryless broadcast channel (DM-BC) which achieve
Marton's region with both common and private messages. This
is the best achievable rate region known to date, and it is
tight for all classes of $2$-user DM-BCs whose capacity regions
are known. To accomplish this task, we first construct polar
codes for both the superposition as well as the binning strategy.
By combining these two schemes, we obtain Marton's region with
private messages only. Finally, we show how to handle the case
of common information. The proposed coding schemes possess the
usual advantages of polar codes, i.e., they have low encoding
and decoding complexity and a super-polynomial decay rate of
the error probability.

We follow the lead of Goela, Abbe, and Gastpar, who recently
introduced polar codes emulating the superposition and binning
schemes. In order to align the polar indices, for both schemes,
their solution involves some degradedness constraints that are
assumed to hold between the auxiliary random variables and the
channel outputs. To remove these constraints, we consider the
transmission of $k$ blocks and employ a chaining construction that
guarantees the proper alignment of the polarized indices.
The techniques described in this work are quite general, and they
can be adopted to many other multi-terminal scenarios whenever
there polar indices need to be aligned.
\end{abstract}



\section{Introduction}
Polar codes, introduced by Ar{\i}kan in \cite{Ari09}, have been demonstrated to
achieve the capacity of any memoryless binary-input output-symmetric channel
with encoding and decoding complexity $\Theta(n \log n)$, where $n$ is
the block length of the code, and a block error probability decaying like
$O(2^{-n^{\beta}})$, for any $\beta\in (0, 1/2)$, under successive cancellation
decoding \cite{ArT09}. A refined analysis of the block error
probability of polar codes leads in \cite{HMTU13} to rate-dependent
upper and lower bounds.

The original point-to-point communication scheme has been
extended, amongst others, to lossless and lossy source coding
\cite{Ar10, KoU09} and to various multi-terminal scenarios, such
as the Gelfand-Pinsker, Wyner-Ziv, and Slepian-Wolf problems
\cite{Kor09thesis, Ar12}, multiple-access channels
\cite{AbT12, TSV12, STYe10, MKLK13, NT13}, broadcast channels
\cite{GAG12, GAG13, GAG13ar}, interference channels \cite{AKV11, WaS14},
degraded relay channels \cite{MK12, ARTKS10},
wiretap channels \cite{ARTKS10, MV11, KG11, HoS10, SaV13},
bidirectional broadcast channels with common and
confidential messages \cite{AS13}, write once memories (WOMs)
\cite{BurshteinS13}, arbitrarily permuted parallel channels
\cite{HSST13}, and multiple description coding \cite{SaP14}.

Goela, Abbe, and Gastpar recently introduced polar coding schemes
for the $m$-user deterministic broadcast channel \cite{GAG12, GAG13ar},
and for the noisy discrete memoryless broadcast channel (DM-BC)
\cite{GAG13, GAG13ar}. For the second scenario, they considered two
fundamental transmission strategies: {\em superposition coding},
in the version proposed by Bergmans \cite{Ber73}, and {\em binning}  \cite{M79}.
In order to guarantee a proper alignment of the polar indices, in both
the superposition and binning schemes, their solution
involves some degradedness constraints that are assumed to hold between
the auxiliary random variables and the channel outputs. It is noted
that two superposition coding schemes were proposed by Bergmans \cite{Ber73}
and Cover \cite{Co72}, and they both achieve the capacity region of the
degraded broadcast channel.
However, it has recently been proven that under MAP decoding, Cover's
strategy always achieves a rate region at least as large as Bergmans',
and this dominance can sometimes be strict \cite{WSBK13}.

In this paper we extend the schemes of \cite{GAG13ar}, and we show how to
achieve Marton's region with both common and private messages.
The original work by Marton \cite{M79} covers the case with only private
messages, and the introduction of common information is due to Gelfand
and Pinsker \cite{GelPin80}. Hence, we will refer to this region as the
Marton-Gelfand-Pinsker (MGP) region (this follows the terminology used,
e.g., in \cite{Liang05thesis, LiaKra07, LiKP11}).
This rate region is tight for all classes of DM-BCs with known capacity
region, and it forms the best inner bound known to date for a $2$-user DM-BC
\cite{kim:nit, Kr07, GGNY14}. Note that it also includes Cover's superposition region.

The crucial point consists in removing the degradedness conditions on
auxiliary random variables and channel outputs\footnote{Note
that, in general, such kind of extra conditions make the achievable
rate region strictly smaller, see \cite{HKU09}.}, in order to achieve
any rate pair inside the region defined by Bergmans' superposition
strategy and by the binning strategy. The ideas which make it possible to lift
the constraints come from recent progress in constructing {\em
universal} polar codes, which are capable of achieving the compound
capacity of the whole class of memoryless binary-input output-symmetric channels
\cite{HRunipol, SaW13}. In short,
first we describe polar codes for the superposition and binning strategies.
Then, by combining these two techniques, we achieve Marton's rate region
with private messages only. Finally, by describing how to transmit common
information, we achieve the whole MGP region.

The current exposition is limited to the case of binary auxiliary random
variables and, only for Bergmans' superposition coding scheme, also to
binary inputs. However, there is no fundamental difficulty in extending
the work to the $q$-ary case (see \cite{NT13, SAT09, MT10, PB13, SaP13}).
The proposed schemes possess the standard
properties of polar codes with respect to encoding and decoding,
which can be performed with complexity $\Theta(n\log n)$, as well
as with respect to the scaling of the block error probability as a function
of the block length, which decays like $O(2^{-n^{\beta}})$ for any
$\beta\in (0, 1/2)$.

The rest of the paper is organized as follows. Section~\ref{sec:rateregions}
reviews the information-theoretic achievable rate regions for DM-BCs and the
rate regions that can be obtained by the polarization-based code constructions
proposed in \cite{GAG13ar},
call them the \GAG constructions.
It proceeds by comparing Bergmans' superposition scheme \cite{Ber73} with the
\GAG superposition region in \cite{GAG13ar}, which serves for motivating this
work. Furthermore, alternative characterizations of superposition,
binning, and Marton's regions are presented in Section~\ref{sec:rateregions}
for simplifying the description of our novel polar coding schemes in this work.
Section~\ref{sec:primitives} reviews two ``polar primitives'' that form the
basis of the \GAG constructions and of our extensions: polar schemes for
lossless compression, with and without side information, and for transmission
over binary asymmetric channels.  Sections~\ref{sec:superposition}
and~\ref{sec:binning} describe our polar coding schemes that achieve the
superposition and binning regions, respectively. Section~\ref{sec:combreg}
first shows polar codes for the achievability of Marton's region with only
private messages and, then, also for the MGP region with both common and
private messages. Section~\ref{sec:conclusion} concludes this paper with
some final thoughts.


\section{Achievable Rate Regions}
\label{sec:rateregions}

\subsection{Information-Theoretic Schemes}
Let us start by considering the rate region that is achievable by
Bergmans' superposition scheme \cite[Theorem~5.1]{kim:nit}, which
provides the capacity region of degraded DM-BCs.

\begin{theorem}[Superposition Region] \label{th:supran}
Consider the transmission over a $2$-user DM-BC $p_{Y_1, Y_2 \mid X}$, where
$X$ denotes the input to the channel, and $Y_1$, $Y_2$ denote the outputs
at the first and second receiver, respectively. Let $V$ be an auxiliary
random variable. Then, for any joint distribution $p_{V, X}$ s.t.
$V - X - (Y_1, Y_2)$ forms a Markov chain, a rate pair $(R_1, R_2)$ is
achievable if
\begin{equation} \label{eq:supran}
\begin{split}
R_1 &< I(X; Y_1 \mid V),\\
R_2 &< I(V; Y_2),\\
R_1+R_2 &< I(X; Y_1).\\
\end{split}
\end{equation}
\end{theorem}

Note that the above only describes a subset of the region actually
achievable by superposition coding.  We get a second subset by
swapping the roles of the two users, i.e., by swapping the indices
$1$ and $2$. The actual achievable region is obtained by the
convex hull of the closure of the union of these two subsets.

The rate region which is achievable by the binning strategy is described
in the following \cite[Theorem~8.3]{kim:nit}:

\begin{theorem}[Binning Region]\label{th:binran}
Consider the transmission over a $2$-user DM-BC $p_{Y_1, Y_2 \mid X}$,
where $X$ denotes the input to the channel, and $Y_1$, $Y_2$ denote
the outputs at the first and second receiver, respectively.
Let $V_1$ and $V_2$ denote auxiliary random variables. Then, for any
joint distribution $p_{V_1, V_2}$ and for any deterministic function
$\phi$ s.t. $X = \phi(V_1,V_2)$, a rate pair $(R_1, R_2)$ is achievable if
\begin{equation}\label{eq:marrand}
\begin{split}
R_1 &< I(V_1; Y_1),\\
R_2 &< I(V_2; Y_2),\\
R_1+R_2 &< I(V_1; Y_1)+I(V_2; Y_2)-I(V_1; V_2).\\
\end{split}
\end{equation}
\end{theorem}

Note that the achievable rate region does not become larger by considering
general distributions $p_{X\mid V_1, V_2}$, i.e., there is no loss of generality
in restricting $X$ to be a deterministic function of $(V_1, V_2)$ (see \cite[Remark~8.4]{kim:nit}).
Furthermore, for deterministic DM-BCs, the choice $V_1 = Y_1$ and $V_2 = Y_2$ in \eqref{eq:marrand}
provides their capacity region (see, e.g., \cite[Example~7.1]{Kr07}).

The rate region in \eqref{eq:marrand} can be enlarged by combining binning
with superposition coding. This leads to Marton's region for a $2$-user
DM-BC where only private messages are available (see \cite[Theorem~2]{M79}
and \cite[Proposition~8.1]{kim:nit}).

\begin{theorem}[Marton's Region]\label{th:bestreg}
Consider the transmission over a $2$-user DM-BC $p_{Y_1, Y_2 \mid X}$, where
$X$ denotes the input to the channel, and $Y_1$, $Y_2$ denote the outputs
at the first and second receiver, respectively. Let $V$, $V_1$, and $V_2$
denote auxiliary random variables. Then, for any joint distribution
$p_{V, V_1, V_2}$ and for any deterministic function $\phi$ s.t.
$X = \phi(V, V_1, V_2)$, a rate pair $(R_1, R_2)$ is achievable if
\begin{equation}\label{eq:best}
\begin{split}
R_1 &< I(V, V_1; Y_1),\\
R_2 &< I(V, V_2; Y_2),\\
R_1+R_2 &< I(V, V_1; Y_1)+I(V_2; Y_2\mid V)-I(V_1; V_2\mid V), \\
R_1+R_2 &< I(V, V_2; Y_2)+I(V_1; Y_1\mid V)-I(V_1; V_2\mid V). \\
\end{split}
\end{equation}
\end{theorem}

Note that the binning region \eqref{eq:marrand} is a special case of
Marton's region \eqref{eq:best} where the random variable $V$ is set
to be a constant. As for the binning region in Theorem \ref{th:binran},
there is no loss of generality in restricting $X$ to be a deterministic
function of $(V, V_1, V_2)$.

In a more general set-up, the users can transmit also common information.
The generalization of Theorem \ref{th:bestreg} to the case with a common
message results in the MGP region.
We denote by $R_0$ the rate associated to the common message, and $R_1$,
$R_2$ continue to indicate the private rates of the first and the second
user, respectively. Then, under the hypotheses of Theorem \ref{th:bestreg},
a rate triple $(R_0, R_1, R_2)$ is achievable if
\begin{equation}\label{eq:bestcom}
\begin{split}
R_0 &< \min\{I(V; Y_1), I(V; Y_2)\},\\
R_0+R_1 &< I(V, V_1; Y_1),\\
R_0+R_2 &< I(V, V_2; Y_2),\\
R_0+R_1+R_2 &< I(V, V_1; Y_1)+I(V_2; Y_2\mid V)-I(V_1; V_2\mid V), \\
R_0+R_1+R_2 &< I(V, V_2; Y_2)+I(V_1; Y_1\mid V)-I(V_1; V_2\mid V).\\
\end{split}
\end{equation}
An equivalent form of this region was derived by Liang
\cite{Liang05thesis, LiaKra07, LiKP11} (see also Theorem~8.4
and Remark~8.6 in \cite{kim:nit}). Note that the MGP
region \eqref{eq:bestcom} is specialized to Marton's region
\eqref{eq:best} when $R_0 = 0$ (i.e., if only private messages exist).
The evaluation of Marton's region in \eqref{eq:best} and the MGP region
in \eqref{eq:bestcom} for DM-BCs has been recently studied in
\cite{GohariA12, GengJNW13, GohariNA_isit2013}, proving also their optimality
for some interesting and non-trivial models of BCs in \cite{GengGNY14, GengN14}.

\subsection{Polar \GAG Constructions}

Let us now compare the results of Theorems \ref{th:supran} and
\ref{th:binran} with the superposition and binning regions that
are achievable by the polarization-based \GAG constructions in
\cite{GAG13ar}.
We write $p \succ q$ to denote that the channel $q$ is stochastically degraded
with respect to the channel $p$.

\begin{theorem}[\GAG Superposition Region]\label{th:AGGsup}
Consider the transmission over a $2$-user DM-BC $p_{Y_1, Y_2 \mid X}$
with a binary input alphabet, where $X$ denotes the input
to the channel, and $Y_1$, $Y_2$ denote the outputs at the first and
second receiver, respectively. Let $V$ be an auxiliary binary random
variable and assume
that $p_{Y_1 \mid V}\succ p_{Y_2 \mid V}$. Then, for any joint
distribution $p_{V, X}$ s.t. $V - X - (Y_1, Y_2)$ forms a Markov
chain and for any rate pair $(R_1, R_2)$ s.t.
\begin{equation}\label{eq:exsup}
\begin{split}
R_1 &< I(X; Y_1 \mid V),\\
R_2 &< I(V; Y_2),\\
\end{split}
\end{equation}
there exists a sequence of polar codes with an increasing block length $n$
that achieves this rate pair with encoding and decoding complexity
$\Theta(n \log n)$, and with a block error probability that decays
like $O(2^{-n^{\beta}})$ for any $\beta \in (0, 1/2)$.
\end{theorem}

\begin{theorem}[\GAG Binning Region]\label{th:AGGbin}
Consider the transmission over a $2$-user DM-BC $p_{Y_1, Y_2\mid X}$,
where $X$ denotes the input to the channel, and $Y_1$, $Y_2$ denote
the outputs at the first and second receiver, respectively. Let $V_1$
and $V_2$ denote auxiliary binary random variables and assume that
$p_{Y_2 \mid V_2}\succ p_{V_1 \mid V_2}$. Then, for any joint
distribution $p_{V_1, V_2}$, for any deterministic function $\phi$ s.t.
$X = \phi(V_1, V_2)$, and for any rate pair $(R_1, R_2)$ s.t.
\begin{equation}\label{eq:AGGbin}
\begin{split}
R_1 &< I(V_1; Y_1),\\
R_2 &< I(V_2; Y_2) -I(V_1; V_2),\\
\end{split}
\end{equation}
there exists a sequence of polar codes with an increasing block length $n$
that achieves this rate pair with encoding and decoding complexity
$\Theta(n \log n)$, and with a block error probability that decays
like $O(2^{-n^{\beta}})$ for any $\beta \in (0, 1/2)$.  \end{theorem}

The rate regions \eqref{eq:exsup} and \eqref{eq:AGGbin} describe
a subset of the regions actually achievable with polar codes by
superposition coding and binning, respectively. However, in some
cases it is not possible to achieve the second subset, since,
by swapping the indices $1$ and $2$, we might not be able to fulfill
the required degradation assumptions.

\subsection{Comparison of Superposition Regions}

As a motivation, before proceeding with the new code constructions and proofs,
let us consider a specific transmission scenario and compare the information-theoretic
superposition region \eqref{eq:supran} and the AGG superposition region \eqref{eq:exsup}
where the latter requires the degradedness assumption $p_{Y_1 \mid V}\succ p_{Y_2 \mid V}$.

In the following, let the channel between $X$ and $Y_1$ be a binary symmetric channel
with crossover probability $p$, namely, a BSC$(p)$, and the channel between $X$ and $Y_2$
be a binary erasure channel with erasure probability $\epsilon$, namely, a BEC$(\epsilon)$.
Let us recall a few known results for this specific model (see \cite[Example~5.4]{kim:nit}).
\begin{enumerate}
\item For any choice of the parameters $p \in (0, 1/2)$ and $\epsilon \in (0, 1)$, the
capacity region of this DM-BC is achieved using superposition coding.
\item For $0 < \epsilon < 2p$, $Y_1$ is a stochastically degraded version of $Y_2$.
\item For $4p(1-p) < \epsilon \le h_2(p)$, $Y_2$ is more capable than $Y_1$, i.e.
$I(X;Y_2)\ge I(X;Y_1)$ for all distributions $p_X$, where
$h_2(p) = -p\log_2 p -(1-p)\log_2 (1-p)$ denotes the binary entropy function.
\end{enumerate}
Let $\mathcal V$ and $\mathcal X$ denote the alphabets of the auxiliary random variable
$V$ and of the input $X$, respectively. Then, if the DM-BC is stochastically degraded
or more capable, the auxiliary random variables satisfy the cardinality bound
$|\mathcal V| \le |\mathcal X|$ \cite{S78}. Consequently, for such a set of parameters,
we can restrict our analysis to binary auxiliary random variables without any loss of
generality. Furthermore, one can assume that the channel from $V$ to $X$ is a BSC, and
that the binary random variable $X$ is symmetric \cite[Lemma 7]{GNSW14}.

First, pick $p= 0.11$ and $\epsilon = 0.2$. In this case, the DM-BC is stochastically
degraded and, as can be seen in Figure~\ref{fig:calcregion2}, the two regions
\eqref{eq:supran} and \eqref{eq:exsup} coincide despite of the presence of the extra
degradedness assumption. In addition, these two regions are non-trivial in the sense
that they improve upon the simple time-sharing scheme in which one user remains silent
and the other employs a point-to-point capacity achieving code.
Then, pick $p= 0.11$ and $\epsilon = 0.4$. In the latter case, the DM-BC is more capable
and, as can be seen in Figure~\ref{fig:calcregion}, the information-theoretic region
\eqref{eq:supran} strictly improves upon the AGG region \eqref{eq:exsup} that coincides
with a trivial time-sharing.

\begin{figure}[tb]
\centering
\subfigure[$\epsilon =0.2$]
{\includegraphics[width=11cm]{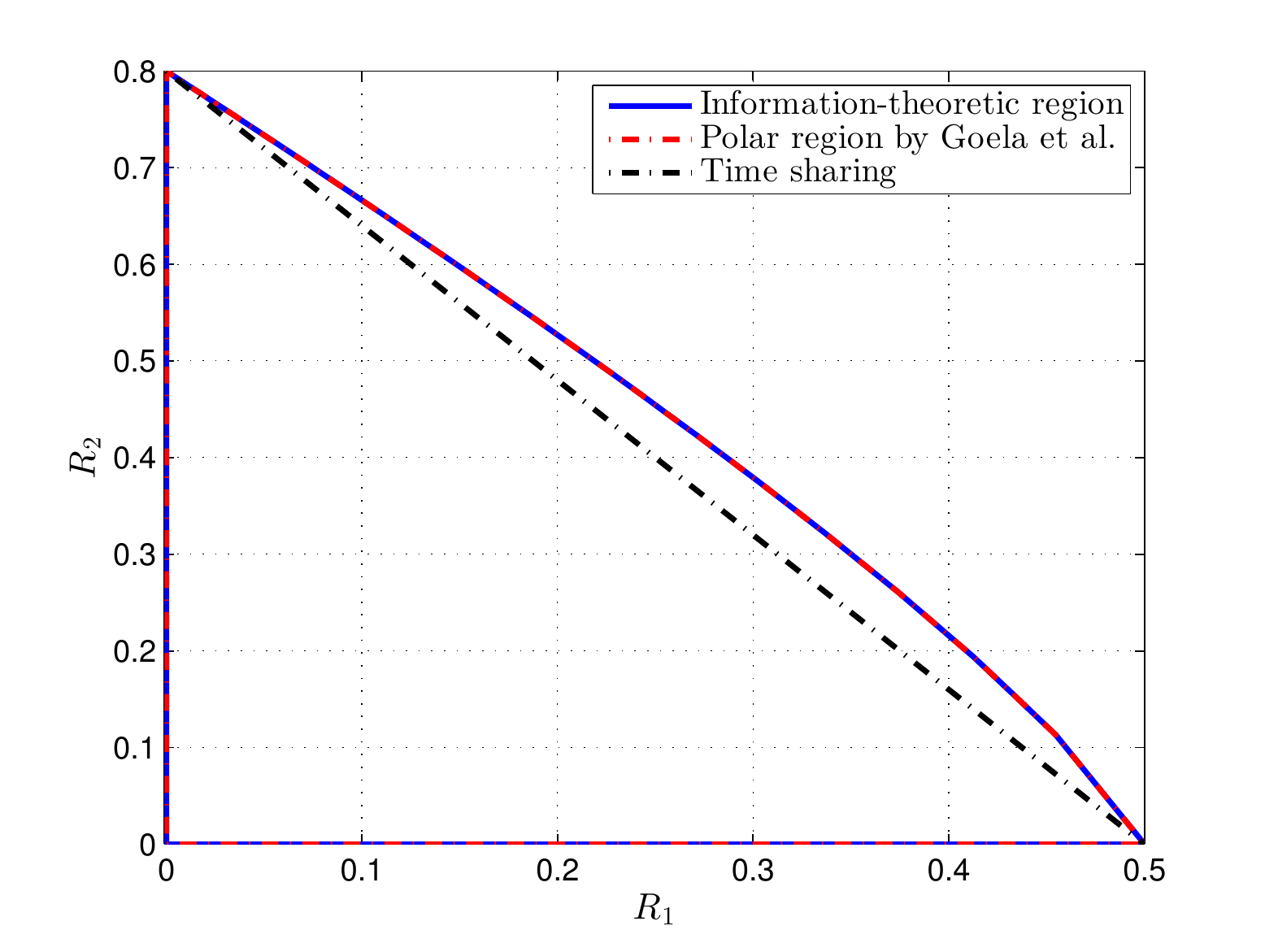}
\label{fig:calcregion2}}
\subfigure[$\epsilon =0.4$]
{\includegraphics[width=11cm]{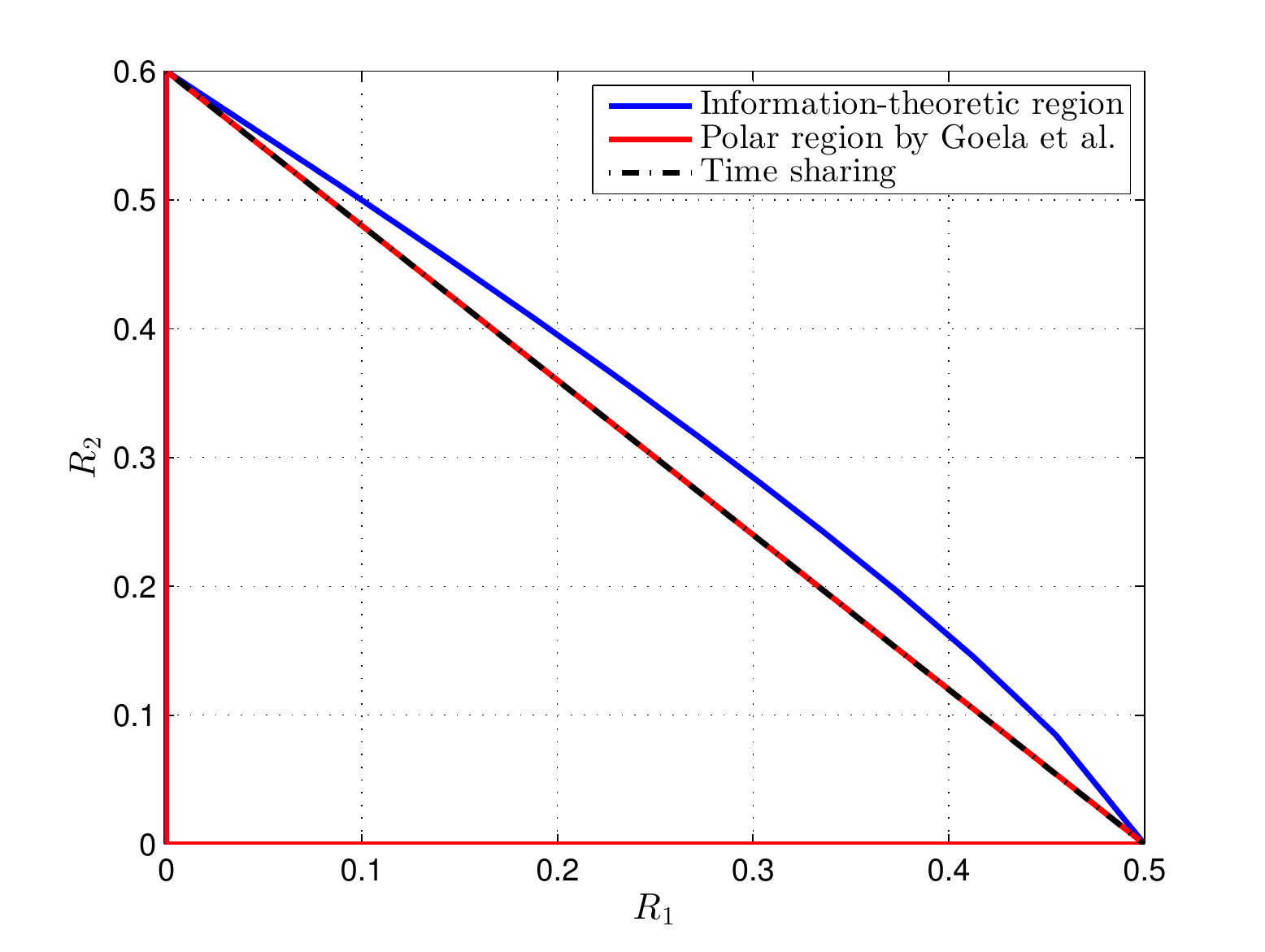}
\label{fig:calcregion}}
\caption{Comparison of superposition regions when the channel from $X$ to $Y_1$ is a BSC$(0.11)$ and the channel from $X$ to $Y_2$ is a BEC$(\epsilon)$. When $\epsilon=0.2$, the information-theoretic region (in blue) coincides with the AGG region (in red) and they are both strictly larger than the time-sharing line (in black). When $\epsilon = 0.4$, the information-theoretic region is strictly larger than the AGG region which reduces to the time-sharing line.}
\end{figure}

\subsection{Equivalent Description of Achievable Regions}\label{subsec:news}
When describing our new polar coding schemes, we will show how to
achieve certain rate pairs. The following propositions state that
the achievability of these rate pairs is equivalent to the achievability
of the whole rate regions described in Theorems~\ref{th:supran}--\ref{th:bestreg}.

\begin{proposition}[Equivalent Superposition Region]\label{prop:sup}
In order to show the achievability of all points in the region \eqref{eq:supran}, it
suffices to describe a sequence of codes with an increasing block length $n$
that achieves each of the rate pairs
\begin{itemize}
\item $(R_1, R_2) = (I(X; Y_1\mid V), \min(I(V; Y_1), I(V; Y_2)))$,
\item $(R_1, R_2) = (I(X; Y_1)- I(V; Y_2), I(V; Y_2))$, provided that
$I(V; Y_1) < I(V; Y_2) < I(X; Y_1)$,
\end{itemize}
with a block error probability that decays to zero as $n \rightarrow \infty$.
\end{proposition}

\begin{proof}
Assume that $I(V; Y_2)\le I(V; Y_1)$. Since $V - X - Y_1$ forms a Markov chain,
by the chain rule, the first two inequalities in \eqref{eq:supran} imply that
$R_1+R_2 < I(X; Y_1\mid V) + I(V; Y_2) \le I(X; Y_1\mid V) + I(V; Y_1) = I(V, X; Y_1) = I(X; Y_1)$.
Hence, the region \eqref{eq:supran} is a rectangle and
it suffices to achieve the corner point $(I(X; Y_1 \mid V), I(V; Y_2))$.

Now, suppose that $I(V; Y_1)< I(V; Y_2)$. Let us separate this case into the
following two sub-cases:
\begin{enumerate}
\item If $I(X; Y_1)> I(V; Y_2)$, the region \eqref{eq:supran} is a pentagon
with the corner points
$$(I(X; Y_1)- I(V; Y_2), I(V; Y_2)), \quad (I(X; Y_1\mid V), I(V; Y_1)).$$
The reason for the first corner point is that
$I(V; Y_1 \mid X) = 0$, so, if $R_2 = I(V; Y_2)$,
the satisfiability of the equality $R_1 + R_2 = I(X; Y_1)$ yields that
$$R_1 = I(X; Y_1)- I(V; Y_2)
= I(V, X; Y_1)- I(V; Y_2) < I(V, X; Y_1)- I(V; Y_1) = I(X; Y_1 \mid V).$$
The reason for the second corner point is that
$R_1 = I(X; Y_1\mid V)$, $R_2 = I(V; Y_1) < I(V; Y_2)$, and
$$R_1 + R_2 = I(VX; Y_1) = I(V; Y_1 \mid X) + I(X; Y_1) = I(X; Y_1).$$
\item Otherwise, if $I(X; Y_1)\leq I(V; Y_2)$, the region \eqref{eq:supran}
is a right trapezoid with corner points $(I(X; Y_1\mid V), I(V; Y_1))$ and
$(0, I(X; Y_1))$. Since $V - X - Y_2$ forms a Markov chain, then, by the data
processing theorem and the last condition, it follows that
$I(X; Y_1)\leq I(V; Y_2) \leq I(X; Y_2)$. Hence, the second corner point
$(0, I(X; Y_1))$ is dominated by the point achievable when the first user
is kept silent and the second user adopts a point-to-point code, taken from a
sequence of codes with an increasing block length $n$, rate close to $I(X; Y_2)$,
and block error probability that decays to zero (for example, a sequence of
polar codes with an increasing block length).
\end{enumerate}
\end{proof}

\begin{proposition}[Equivalent Binning Region]\label{prop:bin}
In order to show the achievability of all points in the region \eqref{eq:marrand},
it suffices to describe a sequence of codes with an increasing block length $n$ that
achieves the rate pair $$(R_1, R_2) = (I(V_1;Y_1), I(V_2;Y_2)-I(V_1;V_2)),$$ assuming that
$I(V_1;V_2) \leq I(V_2;Y_2)$, with a
block error probability that decays to zero as $n \rightarrow \infty$.
\end{proposition}

\bproof
Assume that $I(V_1; V_2) \leq \min(I(V_1; Y_1), I(V_2; Y_2))$. Then, the region
\eqref{eq:marrand} is a pentagon with corner points
$$(I(V_1;Y_1), I(V_2; Y_2)-I(V_1; V_2)), \quad
(I(V_1; Y_1)-I(V_1;V_2), I(V_2;Y_2)).$$ Since the region
\eqref{eq:marrand} and the above condition are not affected
by swapping the indices $1$ and $2$, it suffices to achieve the
first corner point. In order to obtain the other corner point,
one simply exchanges the roles of the two users.

Next, suppose that $I(V_2;Y_2)\le I(V_1; V_2)< I(V_1; Y_1)$.
Then, the region \eqref{eq:marrand} is a right trapezoid with corner points
$$(I(V_1; Y_1)-I(V_1;V_2), I(V_2; Y_2)), \quad
(I(V_1; Y_1)+I(V_2;Y_2)-I(V_1; V_2), 0).$$
Since $I(V_1; Y_1)+I(V_2; Y_2)-I(V_1; V_2)\le
I(V_1; Y_1)$ and $I(V_1; Y_1) \leq I(X; Y_1)$ (this follows
from the data processing theorem for the Markov chain $V_1 - X - Y_1$), the last rate pair is
dominated by the achievable point $(R_1, R_2) = (I(X; Y_1), 0)$ which
refers to a point-to-point communication at rate $I(X; Y_1)$ for the
first user, with a block error probability that decays to zero as
$n \rightarrow \infty$, while the second user is kept silent.

The case where $I(V_1;Y_1)\le I(V_1; V_2)< I(V_2; Y_2)$
is solved by swapping the indices of the two users, and by referring
to the previous case.

Finally, assume that
$I(V_1; V_2)\ge \max(I(V_1; Y_1), I(V_2; Y_2))$.
Then, the region \eqref{eq:marrand} is a triangle with corner
points that are achievable by letting one user remain silent, while the other
user performs a point-to-point reliable communication. \eproof

{\em Remark:} The rate $R_2 = I(V_2; Y_2) - I(V_1; V_2)$ in Proposition~\ref{prop:bin} is
identical to the Gelfand-Pinsker rate if one considers the sequence $V_1^{1:n}$
to be known {\em non-causally} at the encoder. This suggests a design of an encoder
which consists of two encoders: one for $v_1^{1:n}$, and the second for
$v_2^{1:n}$ based on the Gelfand-Pinsker coding; in the second encoder, the
sequence $v_1^{1:n}$ is provided as side information. The reader is referred to the
encoding scheme in \cite[Figure~7.3]{Kr07} while the indices~1 and~2 need to be switched.

\begin{proposition}[Equivalent Marton's Region]\label{prop:combreg}
In order to show the achievability of all points in the region \eqref{eq:best}, it suffices to
describe a sequence of codes with an increasing block length $n$ that achieves each of
the rate pairs
\begin{equation} \label{eq:bestcorner}
\begin{split}
(R_1, R_2) &= (I(V, V_1; Y_1), I(V_2; Y_2\mid V)-I(V_1; V_2\mid V) ),\\
(R_1, R_2) &= (I(V, V_1; Y_1)-I(V_1; V_2\mid V)-I(V; Y_2), I(V, V_2; Y_2) ),\\
\end{split}
\end{equation}
assuming that $I(V; Y_1)\le I(V; Y_2)$, with a block error probability that decays to zero as $n \rightarrow \infty$.
\end{proposition}

\begin{proof}
Since the region \eqref{eq:best} is not affected by swapping the indices $1$ and $2$,
we can assume without loss of generality that $I(V; Y_1)\le I(V; Y_2)$. Then,
\begin{equation*}
\begin{split}
&I(V, V_1; Y_1)+I(V_2; Y_2 \mid V)
=I(V;Y_1) + I(V_1;Y_1|V) + I(V_2; Y_2 \mid V) \\
&\le I(V;Y_2) + I(V_1;Y_1|V) + I(V_2; Y_2 \mid V)= I(V, V_2; Y_2)+I(V_1; Y_1 \mid V),
\end{split}
\end{equation*}
which means that the fourth inequality in \eqref{eq:best} does not restrict the rate region
under the above assumption.

Now, we can follow the same procedure outlined in the proof of Propositions~\ref{prop:sup}
and~\ref{prop:bin}. Suppose that
\begin{equation}\label{ineq:prop}
\begin{split}
&I(V_2; Y_2 \mid V) - I(V_1; V_2 \mid V) > 0,\\
&I(V, V_1; Y_1)-I(V_1; V_2\mid V)-I(V; Y_2) > 0.\\
\end{split}
\end{equation}
Then, the rate region \eqref{eq:best} is a pentagon with the corner points in \eqref{eq:bestcorner}.

If one of the inequalities in \eqref{ineq:prop} is satisfied and the other is violated,
then the region \eqref{eq:best} is a right trapezoid with one corner point given by
\eqref{eq:bestcorner} and the other corner point which is achievable by letting one user remain silent, while the other uses a point-to-point
reliable scheme. If both inequalities in
\eqref{ineq:prop} are violated, then the region \eqref{eq:best} is a triangle with corner points
that are achievable with point-to-point coding schemes.
\end{proof}


\section{Polar Coding Primitives}\label{sec:primitives}

The \GAG constructions, as well as our extensions, are based on two
polar coding ``primitives''.  Therefore, before discussing the
broadcast setting, let us review these basic scenarios.

The first such primitive is the lossless compression, with or
without side information. In the polar setting, this problem was
first discussed in \cite{Kor09thesis, HKU09a}. In Section \ref{subsec:losscompr},
we consider the point of view of source polarization in~\cite{Ar10}.

The second such primitive is the transmission of polar codes over a
general binary-input discrete memoryless channel (a DMC which is either
symmetric or asymmetric). The basic problem which one faces here is that
linear codes impose a uniform input distribution, while the
capacity-achieving input distribution is in general not the uniform
one when the DMC is asymmetric
(however, in relative terms, the degradation in using the uniform
prior for a binary-input DMC is at most 6\% \cite{MajR91, ShuF04}).
One solution consists of concatenating the linear code with a non-linear
pre-mapper \cite{Gal68}. A solution which makes use of the concatenation
of two polar codes has been proposed in \cite{SRDR12}. However, a more
direct polar scheme is implicitly considered in \cite{GAG13ar}, and is
independently and explicitly presented in \cite{HY13}. We will briefly
review this last approach in Section \ref{subsec:general DMC}.

\emph{Notation:} In what follows, we assume that $n$ is a
power of $2$, say $n=2^m$ for $m \in {\mathbb N}$, and we denote
by $G_n$ the polar matrix given by $G_n = \left[\begin{array}{cc}
1 & 0 \\ 1 & 1 \\ \end{array}\right]^{\otimes m}$, where $\otimes$ denotes
the Kronecker product of matrices. The index set $\{1, \cdots, n\}$
is abbreviated as $[n]$ and, given a set $\mathcal A\subseteq [n]$, we denote
by ${\mathcal A}^c$ its complement. We use $X^{i:j}$ as a shorthand
for $(X^i, \cdots, X^j)$ with $i \le j$.

\subsection{Lossless Compression} \label{subsec:losscompr}

{\bf Problem Statement.} Consider a binary random variable $X\sim p_X$.
Then, given the random vector $X^{1:n} = (X^1, \cdots,
X^n)$ consisting of $n$ i.i.d. copies of $X$, the aim is to compress
$X^{1:n}$ in a lossless fashion into a binary codeword of size roughly $n
H(X)$, which is the entropy of $X^{1:n}$.

{\bf Design of the Scheme.}  Let $U^{1:n}= (U^1, \cdots, U^n)$ be
defined as
\begin{equation} \label{lss}
U^{1:n}= X^{1:n} G_n.
\end{equation}
Then, $U^{1:n}$ is a random vector whose components are polarized
in the sense that either $U^i$ is approximately uniform and independent
of $U^{1:i-1}$, or $U^i$ is approximately a deterministic function of
$U^{1:i-1}$. Formally, for $\beta \in (0, 1/2)$, let $\delta_n = 2^{-n^{\beta}}$
and set
\begin{equation} \label{lossless_sets}
\begin{split}
\hset_X &= \{i \in [n] \colon Z(U^i \mid  U^{1:i-1}) \ge 1- \delta_n\}, \\
\lset_X &= \{i \in [n] \colon Z(U^i \mid  U^{1:i-1}) \le \delta_n\},
\end{split}
\end{equation}
where $Z$ denotes the Bhattacharyya parameter. Recall that, given
$(T, V)\sim p_{T, V}$, where $T$ is binary and $V$ takes values in
an arbitrary discrete alphabet ${\mathcal V}$, we define
\begin{equation} \label{eq:Bhattacharyya}
Z(T\mid V) = 2 \sum_{v\in{\mathcal V}}
{\mathbb P}_V(v)\sqrt{{\mathbb P}_{T\mid V}(0\mid v){\mathbb P}_{T\mid V}(1\mid v)}.
\end{equation}
Hence, for  $i \in \mathcal{H}_X$, the bit $U^i$ is approximately
uniformly distributed and independent of the past $U^{1:i-1}$; also, for
$i \in \mathcal{L}_X$, the bit $U^i$ is approximately a deterministic
function of $U^{1:i-1}$. Furthermore,
\begin{equation} \label{eq:card}
\begin{split}
\lim_{n\to \infty} \frac{1}{n} \, | \hset_X|  &= H(X),\\
\lim_{n\to \infty}\frac{1}{n} \, | \lset_X|  &= 1-H(X).\\
\end{split}
\end{equation}
For a graphical representation of this setting, see Figure~\ref{fig:lossless}.

\begin{figure}[tb]
\begin{center} \includegraphics[width=3cm]{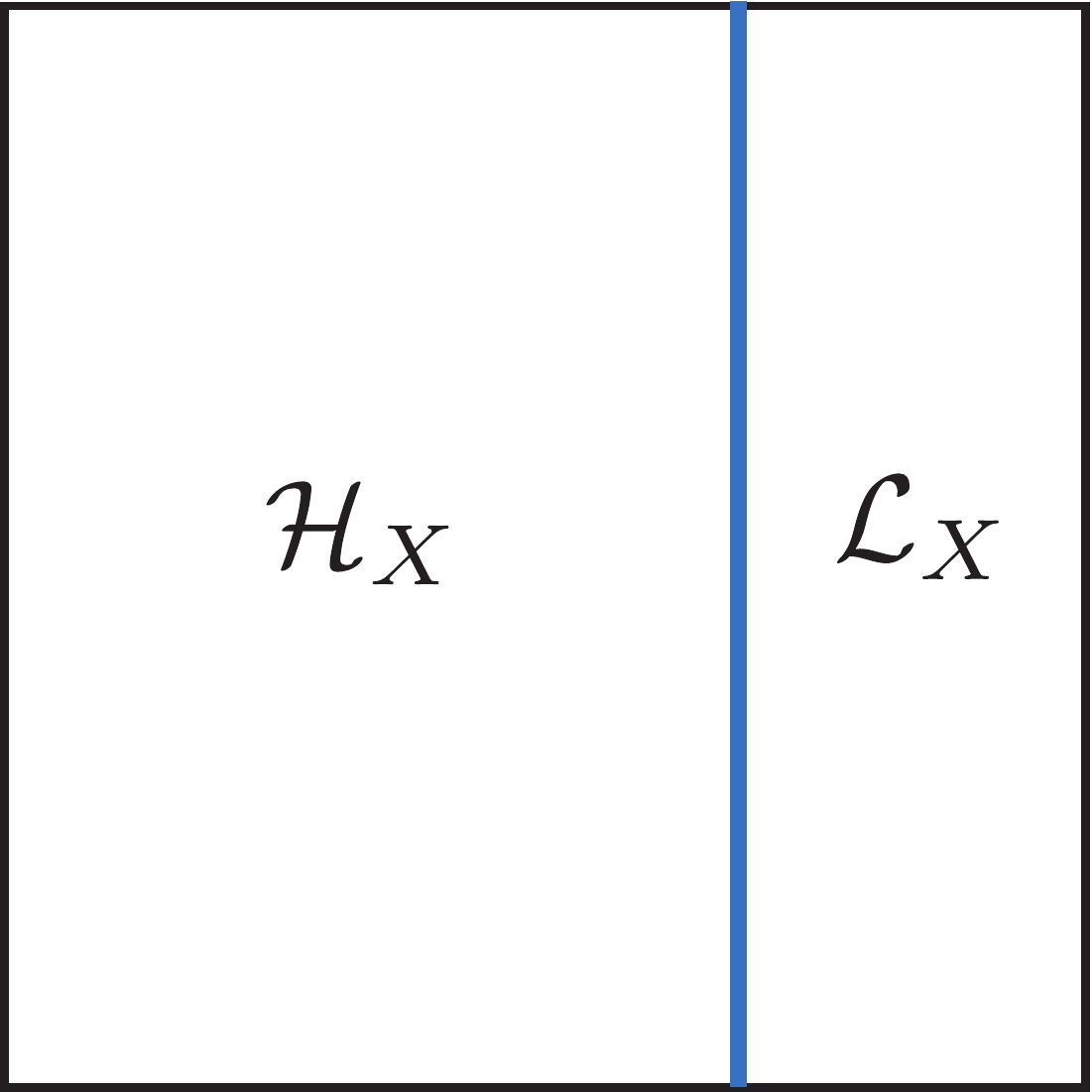}
\end{center}
\caption{A simple graphical representation of the sets $\hset_X$
and $\lset_X$ for the lossless compression scheme.  The whole square represents
$[n]$. The sets $\hset_X$ and $\lset_X$ almost form a
partition of $[n]$ in the sense that the number of indices of $[n]$
which are neither in $\hset_X$ nor in $\lset_X$ is $o(n)$.
}\label{fig:lossless} \end{figure}

{\bf Encoding.} Given the vector $x^{1:n}$ that we want to compress,
the encoder computes $u^{1:n} = x^{1:n} G_n$ and outputs the values
of $u^{1:n}$ in the positions $\lset_X^{\text{c}} = [n] \setminus \lset_X$,
i.e., it outputs $\{u^i\}_{i \in \lset_X^{\text{c}}}$.

{\bf Decoding.} The decoder receives $\{u^i\}_{i \in \lset_X^{\text{c}}}$
and computes an estimate $\hat{u}^{1:n}$ of $u^{1:n}$ using the rule
\begin{equation}\label{eq:decrule} \hat{u}^i = \left\{ \begin{array}{ll}
u^i, & \mbox{if } i \in \lset_X^{\text{c}} \\ \displaystyle\arg\max_{u\in
\{0, 1\}} {\mathbb P}_{U^i \mid  U^{1:i-1}}(u \mid  u^{1:i-1}),& \mbox{if
} i \in \lset_X \\ \end{array}\right..  \end{equation}
Note that the conditional probabilities ${\mathbb P}_{U^i \mid  U^{1:i-1}}(u \mid
u^{1:i-1})$, for $u \in \{0, 1\}$, can be computed recursively with
complexity $\Theta(n \log n)$.

{\bf Performance.}
As explained above, for $i \in \lset_X$, the bit $U^i$ is almost
deterministic given its past $U^{1:i-1}$. Therefore, for $i \in \lset_X$,
the distribution ${\mathbb P}_{U^i \mid  U^{1:i-1}}(u \mid  u^{1:i-1})$
is highly biased towards the correct value $u^i$. Indeed, the block error
probability $P_{\rm e}$, given by
\begin{equation*}
P_{\rm e} = {\mathbb P}(\hat{U}^{1:n} \neq U^{1:n}),
\end{equation*}
can be upper bounded by
\begin{equation}
P_{\rm e} \le \sum_{i \in \lset_X} Z(U^i \mid  U^{1:i-1}) = O(2^{-n^{\beta}}),
\quad \forall \, \beta \in (0, 1/2).
\end{equation}

{\bf Addition of Side Information.} This is a slight extension of the
previous case, and it is also discussed in \cite{Ar10}.
Let $(X, Y) \sim p_{X, Y}$ be a pair of random variables, where we
think of $X$ as the source to be compressed and of $Y$ as a \emph{side
information} about $X$. Given the vector $(X^{1:n}, Y^{1:n})$ of
$n$ independent samples from the distribution $p_{X, Y}$, the
problem is to compress $X^{1:n}$ into a codeword of size roughly
$n H(X\mid Y)$, so that the decoder is able to recover the whole vector
$X^{1:n}$ by using the codeword and the side information $Y^{1:n}$.

Define $U^{1:n}= X^{1:n}
G_n$ and consider the sets
\begin{equation} \label{lossless_side_h}
\hset_{X\mid Y} = \{i \in [n] \colon Z(U^i \mid U^{1:i-1}, Y^{1:n}) \ge 1- \delta_n\},
\end{equation}
representing the positions s.t. $U^i$ is approximately uniformly distributed and
independent of $(U^{1:i-1}, Y^{1:n})$, and
\begin{equation} \label{lossless_side_l}
\lset_{X\mid Y} = \{i \in [n] \colon Z(U^i \mid  U^{1:i-1}, Y^{1:n}) \le \delta_n\},
\end{equation}
representing the positions s.t. $U^i$ is approximately a deterministic
function of $(U^{1:i-1}, Y^{1:n})$ (see Figure~\ref{fig:lossless_side}).
Note that lossless compression without side information can be considered as lossless
compression with side information $\tilde{Y}$, where $\tilde{Y}$ is independent of
$X$ (say, e.g., that $\tilde{Y}$ is constant). Therefore, $\tilde{Y}$ does
not add any information about $X$ and it can be thought as a degraded version of $Y$.
Therefore, the following inclusion relations hold:
\begin{equation}\label{eq:cardnec}
\begin{split}
&\hset_{X \mid Y} \subseteq \hset_X, \\
&\lset_{X} \subseteq \lset_{X \mid Y},
\end{split}
\end{equation}
as it is graphically illustrated in Figures~\ref{fig:lossless} and~\ref{fig:lossless_side}.
A relationship analogous to \eqref{eq:card} holds, namely,
\begin{equation} \label{eq:cardsideinfo}
\begin{split}
\lim_{n\to \infty}\frac{1}{n} \, | \hset_{X\mid Y}|  &= H(X\mid Y),\\[0.1cm]
\lim_{n\to \infty}\frac{1}{n} \, | \lset_{X\mid Y}|  &= 1-H(X\mid Y).
\end{split}
\end{equation}

\begin{figure}[tb]
\begin{center} \includegraphics[width=3cm]{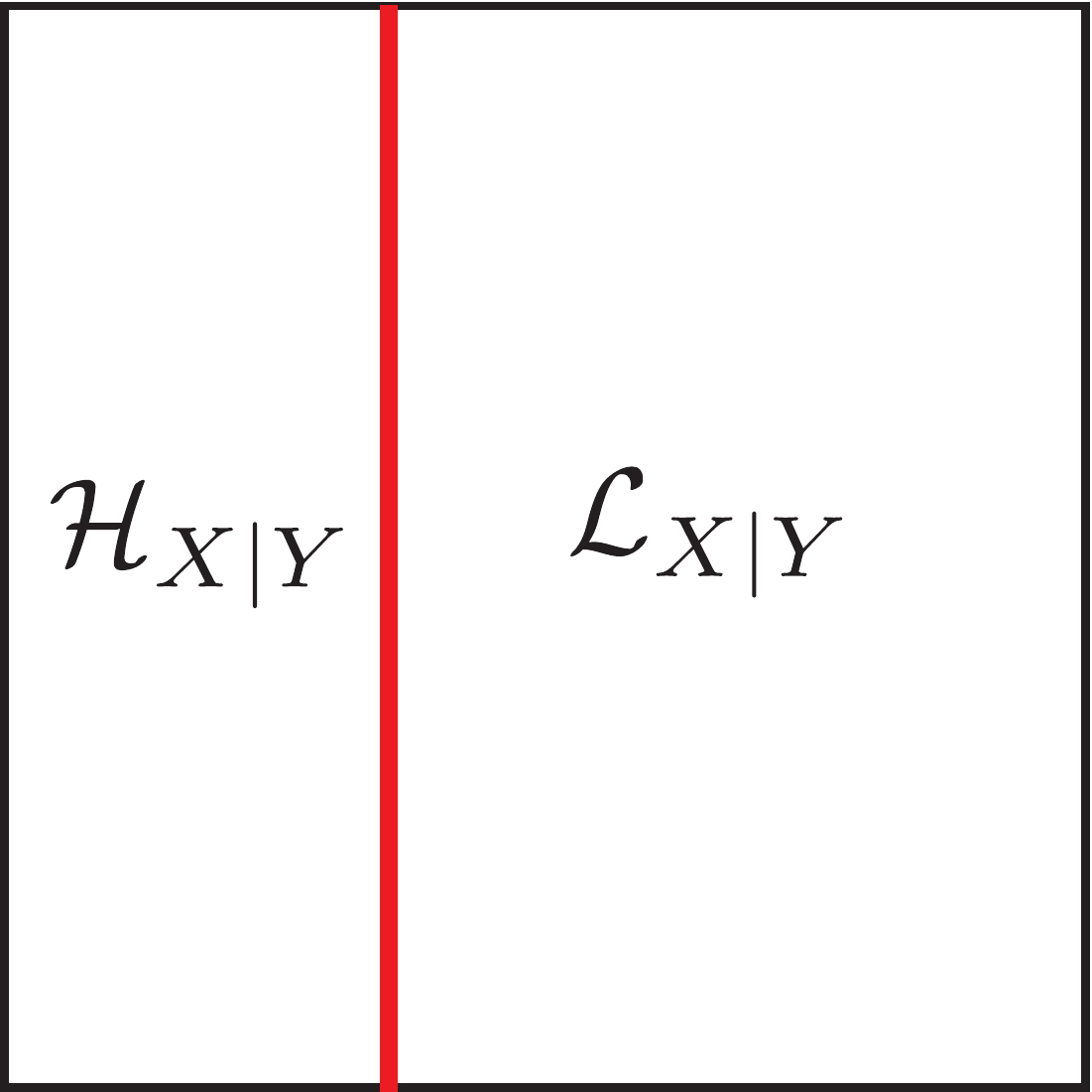}
\end{center}
\caption{A simple graphical representation of the sets $\hset_{X\mid Y}$
and $\lset_{X\mid Y}$ for the lossless compression scheme with side information.  The whole square represents
$[n]$. The sets $\hset_{X\mid Y}$ and $\lset_{X\mid Y}$ almost form a
partition of $[n]$ in the sense that the number of indices of $[n]$
which are neither in $\hset_{X\mid Y}$ nor in $\lset_{X\mid Y}$ is $o(n)$.
}\label{fig:lossless_side}
\end{figure}

Given a realization of $X^{1:n}$, namely $x^{1:n}$, the encoder constructs
$u^{1:n} = x^{1:n} G_n$ and outputs $\{u^i\}_{i \in \lset_{X\mid Y}^{\text{c}}}$ as
the compressed version of $x^{1:n}$. The decoder, using the
side information $y^{1:n}$ and a decoding rule similar to
\eqref{eq:decrule}, is able to reconstruct $x^{1:n}$ reliably with
vanishing block error probability.

\subsection{Transmission over Binary-Input DMCs}
\label{subsec:general DMC}

{\bf Problem Statement.} Let $W$ be a DMC with a binary input $X$ and output $Y$.
Fix a distribution $p_X$ for the random variable $X$.
The aim is to transmit over $W$ with a rate close to $I(X; Y)$.

{\bf Design of the Scheme.}
Let $U^{1:n} = X^{1:n} G_n$, where $X^{1:n}$ is a vector of
$n$ i.i.d. components drawn according to $p_X$.  Consider the sets
$\hset_X$ and $\lset_X$ defined in \eqref{lossless_sets}. From the
discussion about lossless compression, we know that, for $i \in
\hset_X$, the bit $U^i$ is approximately uniformly distributed and independent
of $U^{1:i-1}$ and that, for $i \in \lset_X$, the bit $U^i$ is approximately a
deterministic function of the past $U^{1:i-1}$. Now, assume that the channel
output $Y^{1:n}$ is given, and interpret this as side information on $X^{1:n}$.
Consider the sets $\hset_{X \mid Y}$ and $\lset_{X \mid Y}$ as defined in
\eqref{lossless_side_h} and \eqref{lossless_side_l}, respectively.
To recall, for $i \in \hset_{X \mid Y}$, $U^i$ is approximately uniformly distributed
and independent of $(U^{1:i-1}, Y^{1:n})$, and, for $i \in \lset_{X \mid Y}$, $U^i$
becomes approximately a deterministic function of $(U^{1:i-1}, Y^{1:n})$.

To construct a polar code for the channel $W$, we proceed now as
follows. We place the information in the positions indexed by $\iset
= \hset_X \cap \lset_{X\mid Y}$ (note that, from \eqref{eq:cardnec},
$\lset_{X} \subseteq \lset_{X \mid Y}$). Indeed, if $i \in \iset$, then $U^i$
is approximately uniformly distributed given $U^{1:i-1}$, since $i\in\hset_{X}$.
This implies that $U^i$ is suitable to contain information.
Furthermore, $U^i$ is approximately a deterministic function if we are given
$U^{1:i-1}$ {\em and} $Y^{1:n}$, since $i\in\lset_{X
\mid Y}$.  This implies that it is also decodable in a successive manner
given the channel output. Using \eqref{eq:card}, \eqref{eq:cardnec},
\eqref{eq:cardsideinfo}, and the fact that the number of indices in $[n]$ which
are neither in $\mathcal{H}_X$ nor in $\mathcal{L}_X$ is $o(n)$, it follows that
\begin{equation}\label{eq:I}
\begin{split}
&\lim_{n\to\infty}\frac{1}{n} \, | \iset| \\[0.1cm]
&= \lim_{n\to\infty}\frac{1}{n} \, |\lset_{X \mid Y} \setminus \lset_X| \\[0.1cm]
&= \lim_{n\to\infty}\frac{1}{n} \, |\lset_{X \mid Y}|
- \lim_{n\to\infty}\frac{1}{n} \, |\lset_{X}| \\[0.1cm]
&= H(X)-H(X \mid Y) \\
&= I(X;Y).
\end{split}
\end{equation}
Hence, our requirement on the transmission rate is met.

The remaining positions are frozen. More precisely, they are divided
into two subsets, namely $\fset_{\rm r} = \hset_X \cap \lset_{X \mid
Y}^{\text{c}}$ and $\fset_{\rm d} = \hset_X^{\text{c}}$. For $i\in \fset_{\rm r}$,
$U^i$ is independent of $U^{1:i-1}$, but cannot be reliably decoded using
$Y^{1:n}$.  We fill these positions with bits chosen uniformly
at random, and this randomness is assumed to be shared between the
transmitter and the receiver (i.e., the encoder and the decoder know
the values associated to these positions). For $i \in \fset_{\rm d}$,
the value of $U^i$ has to be chosen in a particular way. This is true since
almost all these positions are in $\lset_X$ and, hence, $U^i$ is approximately
a deterministic function of $U^{1:i-1}$. The situation is schematically
represented in Figure~\ref{fig:assymetric}.

\begin{figure}[tb]
\begin{center} \includegraphics[width=6.5cm]{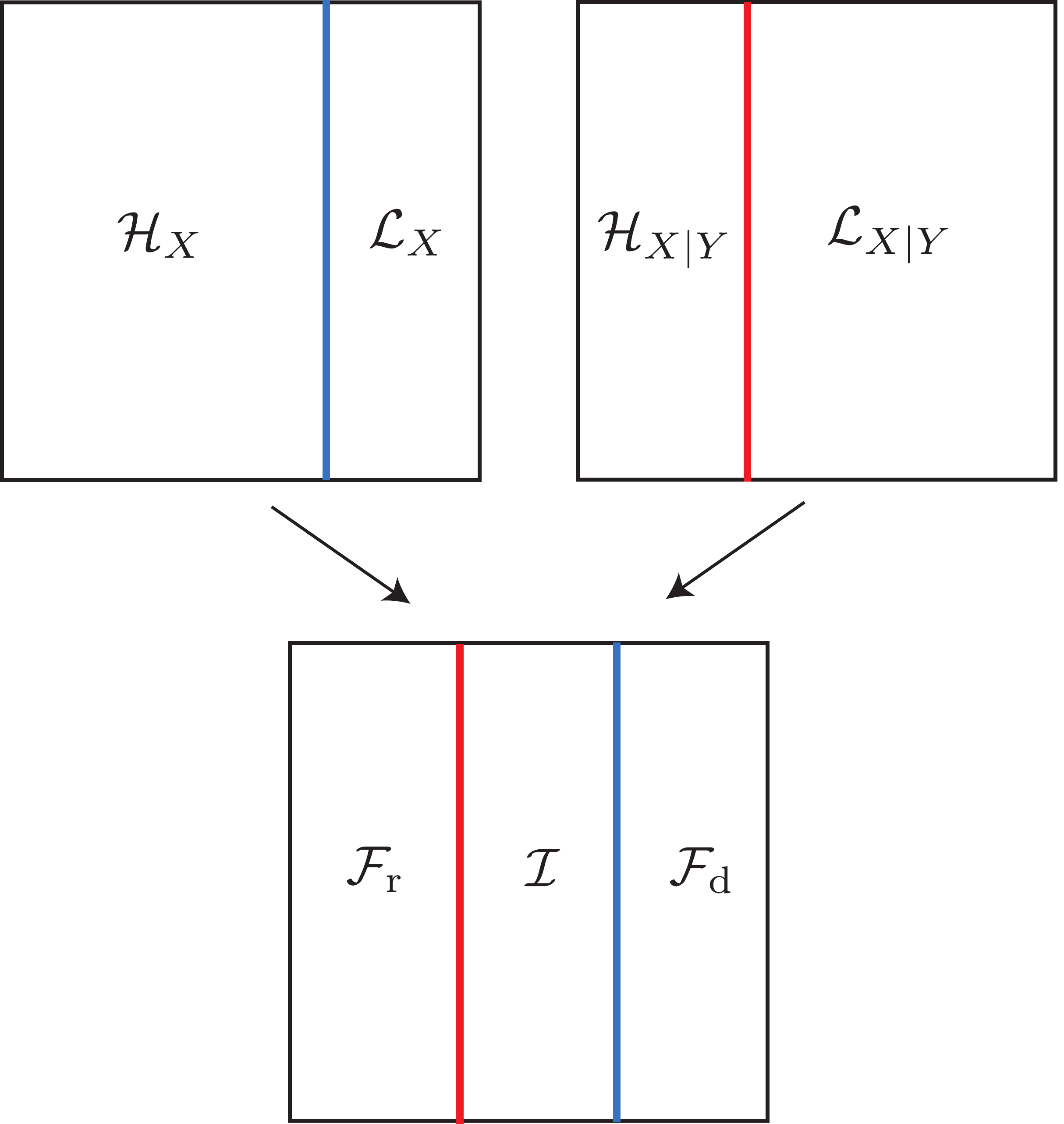}
\end{center}
\caption{Graphical representation of the sets associated
to the channel coding problem. The two images
on top represent how the set $[n]$ (the whole square) is partitioned
by the source $X$ (top left), and by the source $X$ together with the
output $Y$ assumed as a side information (top right). Since $\hset_{X \mid Y}
\subseteq \hset_X$ and $\lset_{X} \subseteq \lset_{X \mid Y}$, the set of
indices $[n]$ can be partitioned into three subsets (bottom image):
the information indices $\mathcal{I} = \hset_X \cap \lset_{X\mid Y}$; the
frozen indices $\mathcal{F}_{\rm r} = \hset_X \cap \lset_{X \mid Y}^{\text{c}}$
filled with binary bits chosen uniformly at random; the frozen indices
$\mathcal{F}_{\rm d} = \hset_X^{\text{c}}$ chosen according to a deterministic rule.}
\label{fig:assymetric} \end{figure}

{\bf Encoding.} The encoder first places the information bits into
$\{u^i\}_{i \in \iset}$. Then, $\{u^i\}_{i \in \fset_{\rm r}}$ is filled
with a random sequence which is shared between the transmitter
and the receiver. Finally, the elements of $\{u^i\}_{i \in \fset_{\rm d}}$
are computed in successive order and, for $i \in \fset_{\rm d}$, $u^i$
is set to the value
$$u^i = \arg\max_{u\in \{0, 1\}} {\mathbb P}_{U^i \mid U^{1:i-1}}(u \mid  u^{1:i-1}).$$
These probabilities can be computed recursively with complexity $\Theta(n \log n)$.
Since $G_n = G_n^{(-1)}$, the $n$-length vector $x^{1:n} = u^{1:n} G_n$ is transmitted
over the channel.

{\bf Decoding.} The decoder receives $y^{1:n}$, and it computes the
estimate $\hat{u}^{1:n}$ of $u^{1:n}$ according to the rule
\begin{equation}\label{eq:decruleas}
\hat{u}^i = \left\{ \begin{array}{ll}
u^i, & \hspace*{0.2em}\mbox{if } i \in \fset_{\rm r} \\
\displaystyle\arg\max_{u\in \{0, 1\}}
{\mathbb P}_{U^i \mid U^{1:i-1}}(u \mid u^{1:i-1}),&\hspace*{0.2em}
\mbox{if } i \in \fset_{\rm d} \\
\displaystyle\arg\max_{u\in \{0, 1\}} {\mathbb P}_{U^i
\mid U^{1:i-1}, Y^{1:n}}(u \mid u^{1:i-1}, y^{1:n}), & \hspace*{0.2em}
\mbox{if } i \in \iset
\\ \end{array}\right.,  \end{equation}
where ${\mathbb P}_{U^i \mid U^{1:i-1}, Y^{1:n}}(u \mid u^{1:i-1}, y^{1:n})$
can be computed recursively with complexity $\Theta(n \log n)$.

{\bf Performance.} The block error probability $P_{\rm e}$ can be upper bounded by
\begin{equation}
P_{\rm e} \le \sum_{i \in \iset} Z(U^i \mid U^{1:i-1}, Y^{1:n}) = O(2^{-n^{\beta}}),
\quad \forall \, \beta \in (0, 1/2).
\end{equation}


\section{Polar Codes for Superposition Region}\label{sec:superposition}

The following theorem provides our main result regarding the achievability
of Bergmans' superposition region for DM-BCs with polar codes (compare with
Theorem~\ref{th:supran}).

\begin{theorem}[Polar Codes for Superposition Region]\label{th:suppolnew}
Consider a $2$-user DM-BC $p_{Y_1, Y_2 \mid X}$ with a binary
input alphabet, where $X$ denotes the input to the channel, and
$Y_1$, $Y_2$ denote the outputs at the first and second receiver, respectively.
Let $V$ be an auxiliary binary random variable. Then, for any joint distribution
$p_{V, X}$ s.t. $V - X - (Y_1, Y_2)$ forms a Markov chain and for any rate pair
$(R_1, R_2)$ satisfying the constraints in \eqref{eq:supran}, there exists a
sequence of polar codes with an increasing block length $n$ which achieves this
rate pair with encoding and decoding complexity $\Theta(n\log n)$ and a block
error probability decaying like $O(2^{-n^{\beta}})$ for any $\beta \in (0, 1/2)$.
\end{theorem}

{\bf Problem Statement.}  Let $(V, X) \sim p_{V, X} = p_{V} p_{X \mid V}$.
We will show how to transmit over the $2$-user DM-BC $p_{Y_1, Y_2 \mid X}$
achieving the rate pair
\begin{equation}\label{eq:rpsup}
(R_1, R_2) = (I(X; Y_1)- I(V; Y_2), I(V; Y_2)),
\end{equation}
when $I(V;Y_1)< I(V; Y_2)< I(X; Y_1)$.
Once we have accomplished this, we will see that a slight
modification of this scheme allows to achieve, in addition, the rate pair
\begin{equation}\label{eq:rpsup2}
(R_1, R_2) = (I(X; Y_1\mid V), \min_{l \in \{1, 2\}} I(V; Y_l)).
\end{equation}
Therefore, by Proposition~\ref{prop:sup}, we can achieve the whole
region~\eqref{eq:supran} and Theorem~\ref{th:suppolnew} is proved.
Note that if polar coding achieves the rate pairs \eqref{eq:rpsup}
and \eqref{eq:rpsup2} with complexity $\Theta(n\log n)$ and a block
error probability $O(2^{-n^{\beta}})$, then for any other rate pair in
the region \eqref{eq:supran}, there exists a sequence of polar codes
with an increasing block length $n$ whose complexity and block
error probability have the same asymptotic scalings.

{\bf Design of the Scheme.} Set $U_2^{1:n} = V^{1:n} G_n$. As in
the case of the transmission over a general binary-input DMC
with $V$ in place of $X$ and $Y_l$ ($l \in \{1, 2\}$) in place of $Y$, define
the sets $\hset_{V}$, $\lset_{V}$, $\hset_{V \mid Y_l}$, and $\lset_{V \mid Y_l}$,
analogously to Section~\ref{subsec:general DMC}, as follows:
\begin{equation} \label{eq:supzero}
\begin{split}
\hset_V &= \{i \in [n] \colon Z(U_2^i \mid  U_2^{1:i-1}) \ge 1- \delta_n\}, \\
\lset_V &= \{i \in [n] \colon Z(U_2^i \mid  U_2^{1:i-1}) \le \delta_n\}, \\
\hset_{V \mid Y_l} &= \{i \in [n] \colon Z(U_2^i \mid  U_2^{1:i-1}, Y_l^{1:n}) \ge 1- \delta_n\},  \\
\lset_{V \mid Y_l} &= \{i \in [n] \colon Z(U_2^i \mid  U_2^{1:i-1}, Y_l^{1:n}) \le \delta_n\},
\end{split}
\end{equation}
which satisfy, for $l \in \{1, 2\}$,
\begin{equation} \label{eq:supfirst}
\begin{split}
\lim_{n\to \infty}\frac{1}{n} \, | \hset_{V}|  &= H(V),\\
\lim_{n\to \infty}\frac{1}{n} \, | \lset_{V}|  &= 1-H(V),\\
\lim_{n\to \infty}\frac{1}{n} \, | \hset_{V\mid Y_l}|  &= H(V \mid Y_l),\\
\lim_{n\to \infty}\frac{1}{n} \, | \lset_{V\mid Y_l}|  &= 1-H(V \mid Y_l).\\
\end{split}
\end{equation}
Set $U_1^{1:n} = X^{1:n} G_n$. By thinking of $V$ as side information on $X$
and by considering the transmission of $X$ over the memoryless channel
with output $Y_1$, define also the sets $\hset_{X\mid V}$, $\lset_{X\mid V}$,
$\hset_{X\mid V, Y_1}$, and $\lset_{X\mid V, Y_1}$, as follows:
\begin{equation}  
\begin{split}
\hset_{X \mid V} &= \{i \in [n] \colon Z(U_1^i \mid  U_1^{1:i-1}, V^{1:n}) \ge 1- \delta_n\}, \\
\lset_{X \mid V} &= \{i \in [n] \colon Z(U_1^i \mid  U_1^{1:i-1}, V^{1:n}) \le \delta_n\}, \\
\hset_{X \mid V, Y_1} &= \{i \in [n] \colon Z(U_1^i \mid  U_1^{1:i-1}, V^{1:n}, Y_1^{1:n}) \ge 1- \delta_n\},  \\
\lset_{X \mid V, Y_1} &= \{i \in [n] \colon Z(U_1^i \mid  U_1^{1:i-1}, V^{1:n}, Y_1^{1:n}) \le \delta_n\},
\end{split}
\end{equation}
which satisfy
\begin{equation} \label{eq:supcod}
\begin{split}
\lim_{n\to \infty}\frac{1}{n} \, | \hset_{X\mid V}|  &= H(X\mid V),\\
\lim_{n\to \infty}\frac{1}{n} \, | \lset_{X\mid V}|  &= 1-H(X\mid V),\\
\lim_{n\to \infty}\frac{1}{n} \, | \hset_{X\mid V,Y_1}|  &= H(X\mid V, Y_1),\\
\lim_{n\to \infty}\frac{1}{n} \, | \lset_{X\mid V, Y_1}|  &= 1-H(X\mid V, Y_1).\\
\end{split}
\end{equation}

First, consider only the point-to-point communication problem between
the transmitter and the second receiver.  As discussed in
Section~\ref{subsec:general DMC}, for this scenario, the correct choice
is to place the information bits in those positions of $U_2^{1:n}$ that
are indexed by the set $\iset^{(2)} = \hset_V \cap \lset_{V\mid
Y_2}$.  If, in addition, we restrict ourselves to positions in
$\iset^{(2)}$ which are contained in $\iset^{(1)}_v = \hset_{V}
\cap \lset_{V\mid Y_1}$, also the first receiver will be able
to decode this message. Indeed, recall that in the superposition
coding scheme, the first receiver needs to decode the message intended
for the second receiver before decoding its own message. Consequently,
for sufficiently large $n$, the first receiver knows the vector $U_2^{1:n}$
with high probability, and, hence, also the vector
$V^{1:n} = U_2^{1:n} G_n$ (recall that $G_n^{-1} = G_n$).

Now, consider the point-to-point communication problem between the
transmitter and the first receiver, given the side information
$V^{1:n}$ (following our discussion, as we let $n$ tend to infinity,
the vector $V^{1:n}$ is known to the first receiver with probability
that tends to~1). From Section~\ref{subsec:general DMC}, we know that
the information has to be placed in those positions of $U_1^{1:n}$ that
are indexed by $\iset^{(1)} = \hset_{X\mid V} \cap \lset_{X\mid V,Y_1}$.

The cardinalities of these information sets are given by
\begin{equation}\label{eq:mutuals} \begin{split}
\lim_{n\to\infty}\frac{1}{n} \, | \iset^{(2)}|  &= I(V;Y_2),\\
\lim_{n\to\infty}\frac{1}{n} \, | \iset^{(1)}_v|  &= I(V;Y_1),\\
\lim_{n\to\infty}\frac{1}{n} \, | \iset^{(1)}|  &= I(X;Y_1\mid V).\\
\end{split} \end{equation}

Let us now get back to the broadcasting scenario, and see how
the previous observations can be used to construct a polar coding
scheme. Recall that $X^{1:n}$ is transmitted over the channel,
the second receiver only decodes its intended message, but the first
receiver decodes both messages.

We start by reviewing the \GAG scheme \cite{GAG13ar}. This scheme
achieves the rate pair
\begin{equation}
(R_1, R_2) = (I(X;Y_1\mid V), I(V;Y_2)),
\end{equation}
assuming that $p_{Y_1\mid V} \succ p_{Y_2\mid V}$. Under this assumption,
we have $\lset_{V\mid Y_2} \subseteq \lset_{V\mid Y_1}$ and therefore
$\iset^{(2)} \subseteq \iset^{(1)}_v$.  Consequently, we can
in fact use the point-to-point solutions outlined above, i.e., the
second user can place his information in $\iset^{(2)}$ and decode,
and the first user will also be able to decode this message. Furthermore,
once the message intended for the second user is known by the first user, the
latter can decode his own information which is placed in the
positions of $\iset^{(1)}$.

Let us now see how to eliminate the restriction imposed by the
degradedness condition $p_{Y_1\mid V} \succ p_{Y_2\mid V}$.
Recall that we want to achieve the rate pair \eqref{eq:rpsup} when
$I(V;Y_1)< I(V; Y_2)< I(X; Y_1)$. The set of indices of the information
bits for the first user is exactly the same as before, namely the positions
of $U_1^{1:n}$ indexed by $\iset^{(1)}$. The only difficulty lies in designing
a coding scheme in which {\em both} receivers can decode the message intended
for the second user.

First of all, observe that we can use all the positions in
$\iset^{(1)}_v \cap \iset^{(2)}$, since they are decodable by both
users.  Let us define
\begin{equation}
\dset^{(2)}=\iset^{(2)} \setminus \iset^{(1)}_v.  \label{eq:dset^2}
\end{equation}
If $p_{Y_1|V} \succ p_{Y_2|V}$, as before, then $\dset^{(2)} = \emptyset$
(i.e., all the positions decodable by the second user are also decodable
by the first user). However, in the general case, where it is no longer
assumed that $p_{Y_1|V} \succ p_{Y_2|V}$, the set $\dset^{(2)}$ is
not empty and those positions cannot be decoded by the first user.

Note that there is a similar set, but with the roles of the two
users exchanged, call it $\dset^{(1)}$, namely,
\begin{equation}
\dset^{(1)} = \iset^{(1)}_v \setminus \iset^{(2)}.  \label{eq:dset^1}
\end{equation}
The set $\dset^{(1)}$ contains the positions of $U_2^{1:n}$ which
are decodable by the first user, but not by the second user. Observe
further that $|\dset^{(1)}| \leq |\dset^{(2)}|$ for sufficiently
large $n$. Indeed, since the equality
\begin{equation} \label{eq:identity for finite sets}
|A \setminus B| - |B \setminus A| = |A|-|B|
\end{equation}
holds for any two finite sets $A$ and $B$, it follows from
\eqref{eq:mutuals}--\eqref{eq:dset^2} that for sufficiently large $n$
\begin{equation}\label{eq:cardB}
\frac{1}{n} \, (|\dset^{(2)}|- |\dset^{(1)}|)
= \frac{1}{n} \, (|\iset^{(2)}|- |\iset^{(1)}_v|)
= I(V; Y_2)-I(V;Y_1)+o(1)\ge 0.
\end{equation}
Assume at first that the two sets are of equal size. The general case will
require only a small modification.

Now, the idea is to consider the ``chaining'' construction introduced
in  \cite{HRunipol} in the context of universal polar codes. Recall that
we are only interested in the message intended for the second user,
but that both receivers must be able to decode this message. Our
scheme consists in transmitting $k$ polar blocks, and in repeating
(``chaining'') some information. More precisely, in block $1$ fill
the positions indexed by $\dset^{(1)}$ with information, but set the bits
indexed by $\dset^{(2)}$ to a fixed known sequence.
In block $j$ ($j \in \{2, \cdots, k-1\}$), fill the positions indexed
by $\dset^{(1)}$ again with information, and repeat the bits which were
contained in the positions indexed by $\dset^{(1)}$ of block $j-1$ into
the positions indexed by $\dset^{(2)}$ of block $j$. In the final block
$k$, put a known sequence in the positions indexed by $\dset^{(1)}$, and
repeat in the positions indexed by $\dset^{(2)}$ the bits in the positions
indexed by $\dset^{(1)}$ of block $k-1$. The remaining bits
are frozen and, as in Section \ref{subsec:general DMC}, they are divided
into the two subsets $\fset_{\rm d}^{(2)} = \hset_{V}^{\text{c}}$ and
$\fset_{\rm r}^{(2)} = \hset_{V} \cap \lset_{V\mid Y_2}^{\text{c}}
\subset \hset_{V}$. In the first case, $U_2^i$
is approximately a deterministic function of $U_2^{1:i-1}$, while in
the second case $U_2^i$ is approximately independent of $U_2^{1:i-1}$.

Note that we lose some rate, since at the boundary we put a known
sequence into some bits which were supposed to contain information.
However, this rate loss decays like $1/k$, and by choosing
a sufficiently large $k$, one can achieve a rate that is arbitrarily
close to the intended rate.

We claim that in the above construction both users can decode all
blocks, but the first receiver has to decode ``forward'', starting
with block $1$ and ending with block $k$, whereas the second receiver
decodes ``backwards'', starting with block $k$ and ending with block
$1$.  Let us discuss this procedure in some more detail. Look at
the first user and start with block $1$. By construction, information
is only contained in the positions indexed by $\dset^{(1)}$ as well
as $\iset^{(1)}_v \cap \iset^{(2)}$, while the positions indexed
by $\dset^{(2)}$ are set to known values. Hence, the first user can
decode this block.  For block $j$ ($j \in \{2, \cdots, k-1\}$), the
situation is similar: the first user decodes the positions indexed
by $\dset^{(1)}$ and $\iset^{(1)}_v \cap \iset^{(2)}$, while the
positions in $\dset^{(2)}$ contain repeated information, which has
been already decoded in the previous block. An analogous analysis
applies to block $k$, in which the positions indexed by
$\dset^{(1)}$ are also fixed to a known sequence. The second user
proceeds exactly in the same fashion, but goes backwards.

To get to the general case, we need to discuss what happens
when $|\dset^{(1)}|< |\dset^{(2)}|$ (due to \eqref{eq:cardB},
in general $|\dset^{(1)}| \leq |\dset^{(2)}|$ for sufficiently
large $n$, but the special case where the two sets are of equal
size has been already addressed). In this case, we do not have
sufficiently many positions in $\dset^{(1)}$ to repeat all the
information contained in $\dset^{(2)}$. To get around this problem,
pick sufficiently many extra positions out of the vector $U_1^{1:n}$
indexed by $\iset^{(1)}$, and repeat the extra information there.

In order to specify this scheme, let us
introduce some notation for the various sets. Recall that we ``chain''
the positions in $\dset^{(1)}$ with an equal amount of positions
in $\dset^{(2)}$. It does not matter what subset of $\dset^{(2)}$
we pick, but call the chosen subset $\rset^{(2)}$. Now, we still
have some positions left in $\dset^{(2)}$, call them $\bset^{(2)}$.
More precisely, $\bset^{(2)}=\dset^{(2)} \setminus \rset^{(2)}$.
Since $\rset^{(2)} \subseteq \dset^{(2)}$ and
$|\rset^{(2)}| = |\dset^{(1)}|$, it follows from \eqref{eq:cardB} that
\begin{equation}\label{eq:rate loss}
\frac{1}{n} \, |\bset^{(2)}|
= \frac{1}{n} \, (|\dset^{(2)}|- |\rset^{(2)}|)
= \frac{1}{n} \, (|\dset^{(2)}|- |\dset^{(1)}|)
= I(V; Y_2)-I(V;Y_1)+o(1)\ge 0.
\end{equation} 
Let $\bset^{(1)}$ be a subset of $\iset^{(1)}$ s.t.
$|\bset^{(1)}|=|\bset^{(2)}|$. Again, it does not matter what
subset we pick. The existence of such a set $\bset^{(1)}$, for sufficiently
large $n$, is ensured by noticing that from \eqref{eq:mutuals},
\eqref{eq:rate loss} and the Markovity of the chain $V - X - Y_1$ we obtain
\begin{equation}\label{eq:used for R1}
\frac{1}{n} \, (|\iset^{(1)}|-|\bset^{(2)}|)
= I(X; Y_1 \mid V) - I(V; Y_2) + I(V; Y_1) + o(1)= I(X; Y_1) - I(V; Y_2) + o(1) \ge 0.
\end{equation}
Indeed, recall that we need to achieve the rate pair \eqref{eq:rpsup} when
$I(V;Y_1)< I(V; Y_2)< I(X; Y_1)$.

As explained above, we place in $\bset^{(1)}$ the value of those
extra bits from $\dset^{(2)}$ which will help the first user to
decode the message of the second user in the next block.
Operationally, we repeat the information contained in the positions
indexed by $\bset^{(2)}$ into the positions indexed by $\bset^{(1)}$
of the previous block. By doing this, the first user pays a rate penalty of
$I(V; Y_2)- I(V; Y_1)+o(1)$ compared to his original rate given by
$\frac{1}{n} \, |\iset^{(1)}|=I(X; Y_1 |V)+o(1)$.

To summarize, the first user puts information bits at positions
$\iset^{(1)}\setminus \bset^{(1)}$, repeats in $\bset^{(1)}$ the
information bits in $\bset^{(2)}$ for the next block, and freezes
the rest. In the last block, the information set is the whole
$\iset^{(1)}$. The frozen positions are divided into the usual two
subsets $\fset_{\rm r}^{(1)} = \hset_{X\mid  V} \cap \lset_{X\mid
V, Y_1}^{\text{c}}$ and $\fset_{\rm d}^{(1)} = \hset_{X\mid V}^{\text{c}}$, which
contain positions s.t. $U_1^i$ is or is not, respectively, approximately
independent of $(U_1^{1:i-1}, V^{1:n})$. The situation is schematically
represented in Figures~\ref{fig:sp_user1}--\ref{fig:sp_chain}.

Suppose that, by applying the same scheme with $k \rightarrow \infty$,
we let $\frac{1}{n} \, |\bset^{(2)}|$ shrink from $I(V; Y_2)-I(V; Y_1)+o(1)$ in
\eqref{eq:rate loss} to $o(1)$. Then, one obtains the whole line going from the
rate pair $(I(X; Y_1)- I(V; Y_2), I(V; Y_2))$ to
$(I(X; Y_1\mid V), I(V; Y_1))$ without time-sharing.\footnote{The reader will be able
to verify this property by relying on \eqref{eq:R2} and \eqref{eq:R1}; this
property is mentioned, however, at this stage as part of the exposition of the
polar coding scheme.}

Finally, in order to obtain the rate pair $(I(X; Y_1\mid V), I(V; Y_2))$
when $I(V;Y_2)\le I(V; Y_1)$, it suffices to consider the case where
$\bset^{(2)} = \emptyset$ and switch the roles of $\iset^{(2)}$ and
$\iset^{(1)}_v$ in the discussion concerning the second user.

\begin{figure}[tb]
\begin{center} \includegraphics[width=6.5cm]{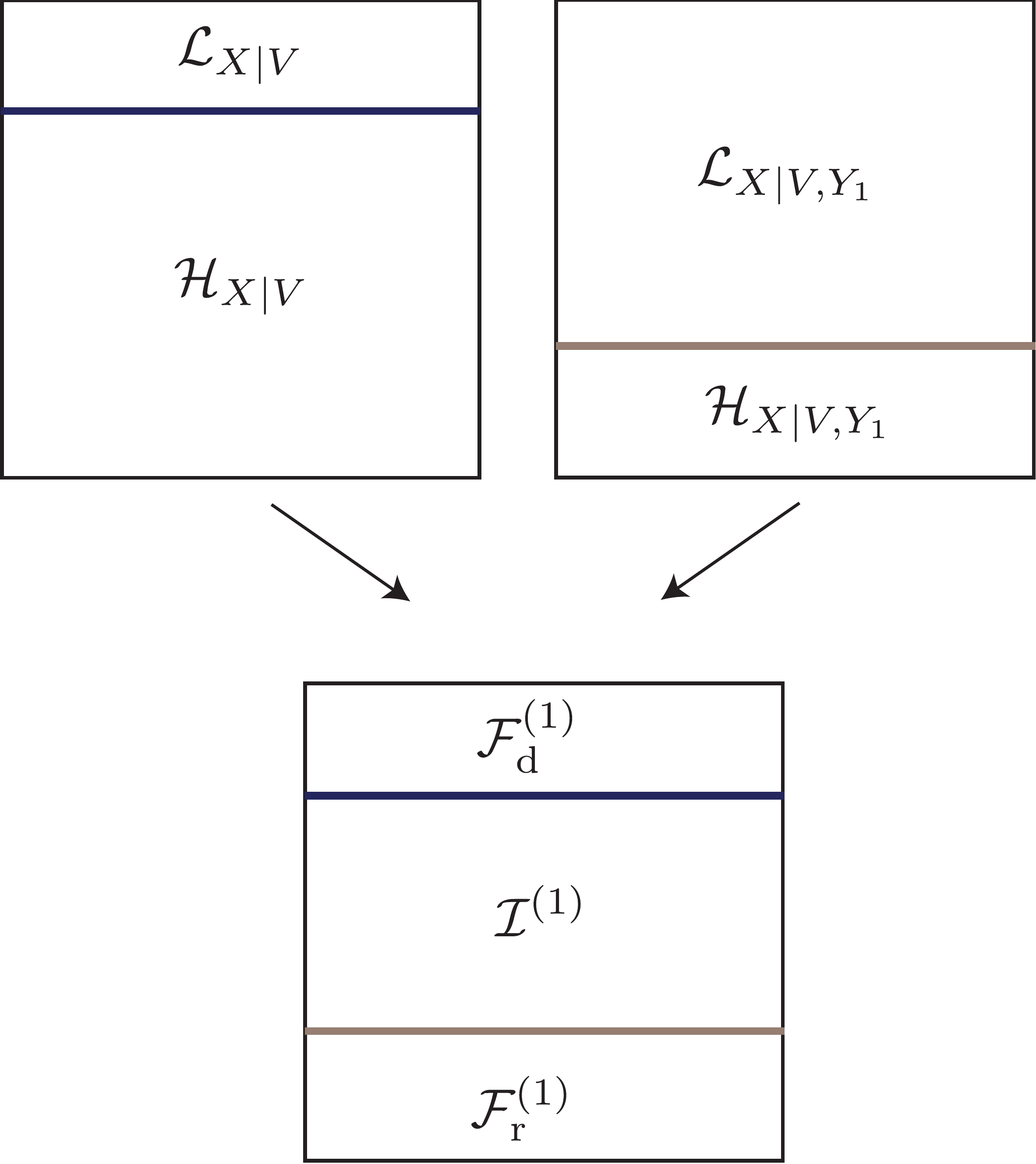}
\end{center}
\caption{Graphical representation of the sets associated to the
first user for the superposition scheme. The set $[n]$ is
partitioned into three subsets: the information indices $\iset^{(1)}$;
the frozen indices $\fset^{(1)}_{\rm r}$ filled with bits
chosen uniformly at random; the frozen indices $\fset^{(1)}_{\rm
d}$ chosen according to a deterministic rule.}\label{fig:sp_user1}
\end{figure}

\begin{figure}[tb]
\begin{center} \includegraphics[width=9.5cm]{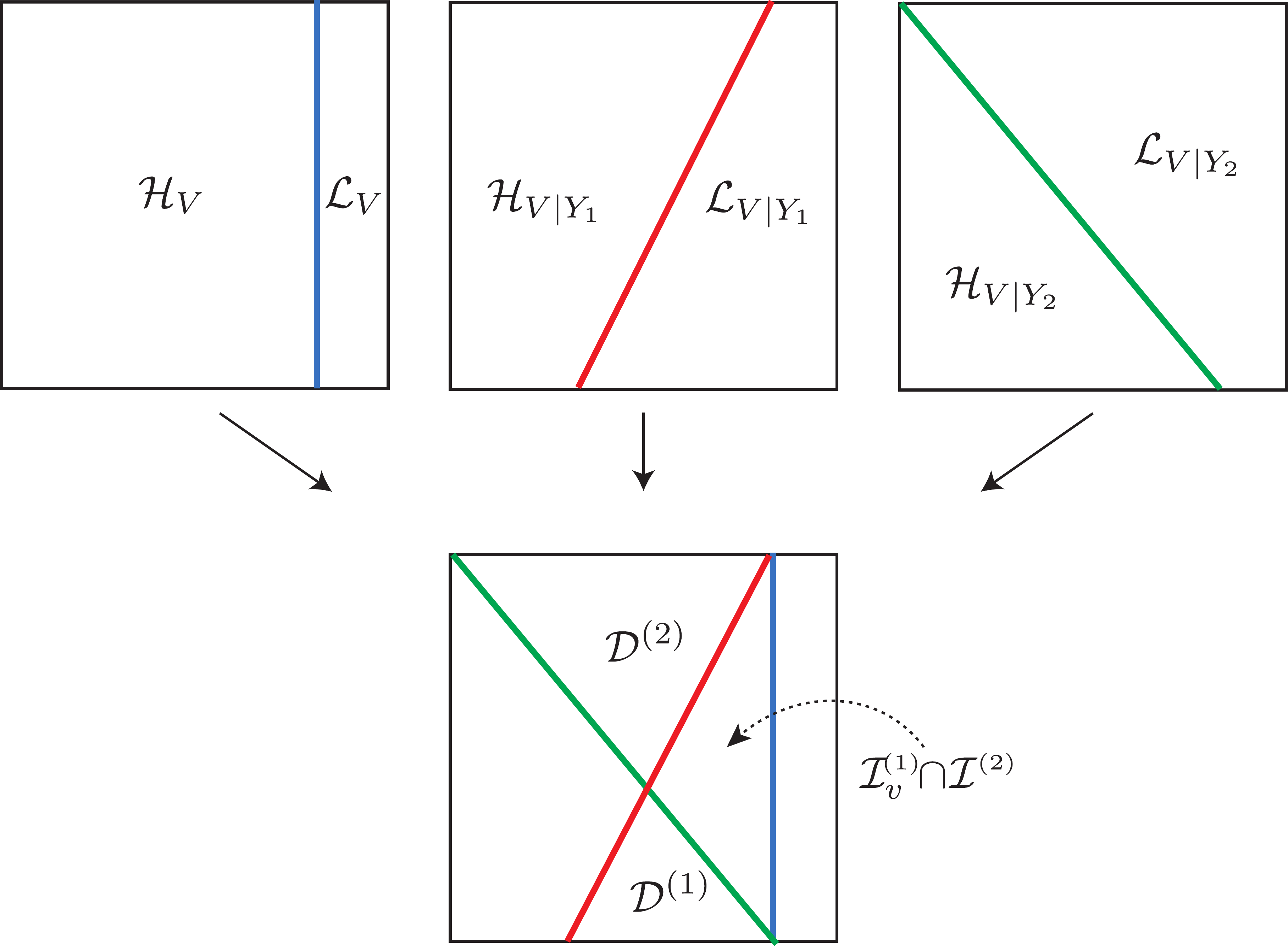}
\end{center}
\caption{Graphical representation of the sets associated to the
second user for the superposition scheme: $\iset^{(1)}_v \cap
\iset^{(2)}$ contains the indices which are decodable by both users;
$\dset^{(1)}=\iset^{(1)}_v \setminus \iset^{(2)}$ contains the indices
which are decodable by the first user, but not by the second user;
$\dset^{(2)}=\iset^{(2)} \setminus \iset^{(1)}_v$ contains the indices
which are decodable by the second user, but not by the first user.}
\label{fig:sp_user2} \end{figure}

\begin{figure}[tb]
\begin{center} \includegraphics[width=9.5cm]{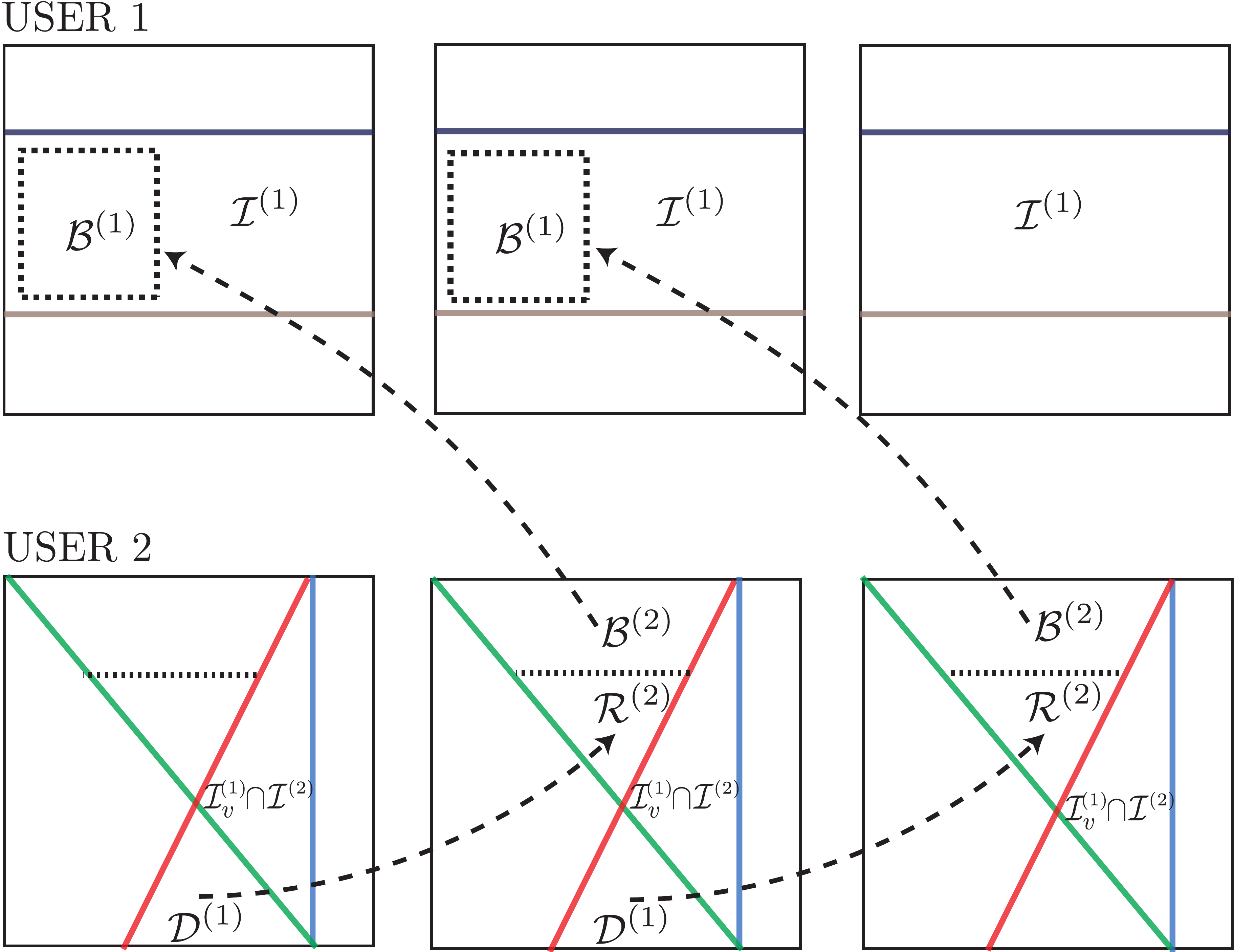}
\end{center}
\caption{Graphical representation of the repetition construction
for the superposition scheme with $k=3$: the set $\dset^{(1)}$
is repeated into the set $\rset^{(2)}$ of the following block; the
set $\bset^{(2)}$ is repeated into the set $\bset^{(1)}$ of the
previous block (belonging to a different user).}\label{fig:sp_chain}
\end{figure}

{\bf Encoding.} Let us start from the second user, and encode block by block.
\newline For block $1$:
\begin{itemize}
\item The information bits are stored in $\{u_2^i\}_{i\in\iset^{(1)}_v}$.
\item The set $\{u_2^i\}_{i\in \fset_{\rm r}^{(2)}}$ is filled with a random sequence,
shared between the transmitter and both receivers.
\item For $i\in \fset_{\rm d}^{(2)}$, we set
$u_2^i = \arg\max_{u\in \{0, 1\}} {\mathbb P}_{U_2^i\mid U_2^{1:i-1}}(u \mid  u_2^{1:i-1})$.
\end{itemize}
For block $j$ ($j \in \{2, \cdots, k-1\})$:
\begin{itemize}
\item The information bits are stored in
$\{u_2^i\}_{i\in\iset^{(1)}_v\cup \bset^{(2)}}$.
\item $\{u_2^i\}_{i\in \rset^{(2)}}$
contains the set $\{u_2^i\}_{i\in \dset^{(1)}}$ of block $j-1$.
\item The frozen sets
$\{u_2^i\}_{i\in \fset_{\rm r}^{(2)}}$ and $\{u_2^i\}_{i\in \fset_{\rm d}^{(2)}}$
are chosen as in block $1$.
\end{itemize}
For block $k$ (the last one):
\begin{itemize}
\item The information bits are stored in
$\{u_2^i\}_{i\in(\iset^{(1)}_v\cap \iset^{(2)})\cup \bset^{(2)}}$.
\item $\{u_2^i\}_{i\in \rset^{(2)}}$ contains the set $\{u_2^i\}_{i\in \dset^{(1)}}$ of block $k-1$.
\item The frozen bits are computed with the usual rules.
\end{itemize}
The rate of the second user is given by
\begin{equation} \label{eq:R2}
\begin{split}
R_2 &= \frac{1}{kn} \, \Bigl[ \bigl| \iset^{(1)}_v \bigr| +(k-2)
\bigl| \iset^{(1)}_v\cup\bset^{(2)} \bigr| + \bigl|
(\iset^{(1)}_v\cap\iset^{(2)})\cup\bset^{(2)} \bigr| \Bigr] \\[0.1cm]
&=\left(\frac{k-1}{k}\right) I(V; Y_2) + \frac{1}{kn}
\, | \iset^{(1)}_v\cap\iset^{(2)}|+o(1),
\end{split}
\end{equation}
which, as $k$ tends to infinity, approaches the required rate $I(V; Y_2)$ (the
second equality in \eqref{eq:R2} follows from \eqref{eq:mutuals} and \eqref{eq:rate loss},
and from the fact that the sets $\iset^{(1)}_v$ and $\bset^{(2)}$ are disjoint).
Then, the vector $v^{1:n} = u_2^{1:n} G_n$ is obtained.

The encoder for the first user knows $v^{1:n}$ and proceeds block by block:
\begin{itemize}
\item The information bits are stored in $\{u_1^i\}_{i\in\iset^{(1)}\setminus \bset^{(1)}}$,
except for block $k$, in which the information set is $\{u_1^i\}_{i\in\iset^{(1)}}$.
\item For block $j$ ($j\in \{1, \cdots, k-1\}$), $\{u_1^i\}_{i\in\bset^{(1)}}$
contains a copy of the set $\{u_2^i\}_{i\in\bset^{(2)}}$ in block $j+1$.
\item The frozen set $\{u_1^i\}_{i\in \fset_{\rm r}^{(1)}}$
contains a random sequence shared between the encoder and the first decoder.
\item For $i\in \fset_{\rm d}^{(1)}$, we set
$u_1^i = \arg\max_{u\in \{0, 1\}} {\mathbb P}_{U_1^i \mid U_1^{1:i-1}, V^{1:n}}(u \mid  u_1^{1:i-1}, v^{1:n})$.
\end{itemize}
The rate of the first user is given by (see \eqref{eq:mutuals} and \eqref{eq:used for R1},
and recall that $\bset^{(1)}\subset \iset^{(1)}$ s.t. $|\bset^{(1)}|=|\bset^{(2)}|$)
\begin{equation} \label{eq:R1}
\begin{split}
R_1 &= \frac{1}{kn} \, \Bigl[ (k-1)| \iset^{(1)}\setminus \bset^{(1)}|  + | \iset^{(1)}|  \Bigr] \\[0.1cm]
&= I(X; Y_1\mid V)-\frac{k-1}{k} \, \bigl(I(V; Y_2)-I(V; Y_1) \bigr)+o(1),
\end{split}
\end{equation}
which, as $k$ tends to infinity, approaches the required rate $I(X; Y_1)-I(V; Y_2)$.
Finally, the vector $x^{1:n} = u_1^{1:n}G_n$ is transmitted over the channel. The
encoding complexity per block is $\Theta(n \log n)$.

{\bf Decoding.} Let us start from the first user, which receives the channel
output $y_1^{1:n}$. The decoder acts block by block and reconstructs first
$u_2^{1:n}$, computes $v^{1:n} = u_2^{1:n}G_n$, and then decodes $u_1^{1:n}$,
thus recovering his own message.
\newline For block $1$, the decision rule is given by
\begin{equation} \label{eq:S12}
\hat{u}_2^i = \left\{ \begin{array}{ll}  u_2^i, &
\mbox{if } i \in \fset_{\rm r}^{(2)} \\
\displaystyle\arg\max_{u\in \{0, 1\}} {\mathbb P}_{U_2^i \mid
U_2^{1:i-1}}(u \mid u_2^{1:i-1}),
& \mbox{if } i \in \fset_{\rm d}^{(2)} \\
\displaystyle\arg\max_{u\in \{0, 1\}} {\mathbb P}_{U_2^i \mid
U_2^{1:i-1}, Y_1^{1:n}}(u \mid u_2^{1:i-1}, y_1^{1:n}),
& \mbox{if } i \in \iset^{(1)}_v \\ \end{array}\right.,
\end{equation}
and
\begin{equation}\label{eq:S1}
\hat{u}_1^i = \left\{ \begin{array}{ll} u_1^i,
& \mbox{if } i \in \fset_{\rm r}^{(1)} \\
\displaystyle\arg\max_{u\in \{0, 1\}}
{\mathbb P}_{U_1^i \mid  U_1^{1:i-1}, V^{1:n}}(u \mid u_1^{1:i-1}, v^{1:n}),
& \mbox{if } i \in \fset_{\rm d}^{(1)} \\
\displaystyle\arg\max_{u\in \{0, 1\}}
{\mathbb P}_{U_1^i \mid  U_1^{1:i-1}, V^{1:n}, Y_1^{1:n}}(u
\mid u_1^{1:i-1}, v^{1:n}, y_1^{1:n}),
 & \mbox{if } i \in \iset^{(1)} \end{array}\right..
\end{equation}
For block $j$ ($j\in \{2, \cdots, k-1\}$):
\begin{itemize}
\item $\{\hat{u}_2^i\}_{i\in \bset^{(2)}}$
is deduced from $\{\hat{u}_1^i\}_{i\in \bset^{(1)}}$ of block $j-1$.
\item $\{\hat{u}_2^i\}_{i\in \rset^{(2)}}$ is deduced from $\{\hat{u}_2^i\}_{i\in \dset^{(1)}}$
of block $j-1$.
\item For the remaining positions of $\hat{u}_2^i$, the decoding
follows the rule in \eqref{eq:S12}.
\item The decoding of $\hat{u}_1^i$ proceeds as in
\eqref{eq:S1}.
\end{itemize}
This decoding rule works also for block~$k$, with the only difference
that the frozen set $\fset_{\rm r}^{(2)}$ is bigger, and
$\hat{u}_2^i=\arg\max_{u\in \{0, 1\}} {\mathbb P}_{U_2^i \mid U_2^{1:i-1},
Y_1^{1:n}}(u \mid u_2^{1:i-1}, y_1^{1:n})$ only for $i\in \iset^{(1)}_v\cap \iset^{(2)}$.

Let us consider now the second user, which reconstructs $u_2^{1:n}$ from the
channel output $y_2^{1:n}$. As explained before, the decoding goes ``backwards'',
starting from block $k$ and ending with block $1$.
\newline For block $k$, the decision rule is given by
\begin{equation} \label{eq:S2}
\hat{u}_2^i = \left\{ \begin{array}{ll} u_2^i,
& \mbox{if } i \in \fset_{\rm r}^{(2)} \\
\displaystyle\arg\max_{u\in \{0, 1\}}
{\mathbb P}_{U_2^i \mid  U_2^{1:i-1}}(u \mid u_2^{1:i-1}),
& \mbox{if } i \in \fset_{\rm d}^{(2)} \\
\displaystyle\arg\max_{u\in \{0, 1\}}
{\mathbb P}_{U_2^i \mid  U_2^{1:i-1}, Y_2^{1:n}}(u \mid u_2^{1:i-1},
y_2^{1:n}), &\mbox{if } i \in (\iset^{(1)}\cap \iset^{(2)})\cup
\rset^{(2)}\cup \bset^{(2)} \\ \end{array}\right..
\end{equation}
For block $j$ ($j\in \{2, \cdots, k-1\}$), the decoder recovers
$\{u_2^i\}_{i\in \dset^{(1)}}$ from $\{u_2^i\}_{i\in \rset^{(2)}}$
of block $j+1$; for the remaining positions, the decision rule in
\eqref{eq:S2} is used.
\newline For block $1$, the reasoning is the same, except that the information bits are
$\{u_2^i\}_{i\in \iset^{(1)}_v\cap \iset^{(2)}}$, i.e., the information
set is smaller. The complexity per block, under successive cancellation decoding, is $\Theta(n\log n)$.

{\bf Performance.} The block error probability $P_{\rm e}^{(l)}$ for
the $l$-th user ($l\in\{1, 2\}$) can be upper bounded by
\begin{equation}
\begin{split}
P_{\rm e}^{(1)} &\le k \sum_{i \in \iset^{(1)}_v} Z(U_2^i
\mid U_2^{1:i-1}, Y_1^{1:n})
+ k \sum_{i \in \iset^{(1)}} Z(U_1^i \mid U_1^{1:i-1},
Y_1^{1:n}) = O(2^{-n^{\beta}}),\\[0.1cm]
P_{\rm e}^{(2)} &\le k \sum_{i \in \iset^{(2)}} Z(U_2^i
\mid U_2^{1:i-1}, Y_2^{1:n}) = O(2^{-n^{\beta}}), \quad \forall\,\beta \in (0, 1/2).
\end{split}
\end{equation}


\section{Polar Codes for Binning Region}\label{sec:binning}

The following theorem provides our main result regarding the
achievability of the binning region for DM-BCs
with polar codes (compare with Theorem~\ref{th:binran}).

\begin{theorem}[Polar Codes for Binning Region] \label{th:polbinnew}
Consider a $2$-user DM-BC $p_{Y_1, Y_2\mid X}$, where $X$ denotes the
input to the channel taking values on an arbitrary set $\mathcal{X}$,
and $Y_1$, $Y_2$ denote the outputs at the first and second receiver,
respectively.
Let $V_1$ and $V_2$ denote auxiliary binary random variables. Then, for
any joint distribution $p_{V_1, V_2}$, for any deterministic function
$\phi \colon \{0, 1\}^2 \rightarrow \mathcal{X}$ s.t.
$X = \phi(V_1, V_2)$, and for any rate pair $(R_1, R_2)$ satisfying
the constraints \eqref{eq:marrand}, there exists a sequence of polar
codes with an increasing block length $n$ which achieves this rate pair
with encoding and decoding complexity $\Theta(n \log n)$ and a block
error probability decaying like $O(2^{-n^{\beta}})$ for any
$\beta \in (0, 1/2)$.  \end{theorem}

{\bf Problem Statement.} Let $(V_1, V_2) \sim p_{V_1, V_2} =
p_{V_1}p_{V_2\mid  V_1}$, and let $X$ be a deterministic function
$\phi$ of $(V_1, V_2)$. The aim is to transmit over the $2$-user
DM-BC $p_{Y_1, Y_2 \mid X}$ achieving the rate pair
\begin{equation}\label{eq:rpmar}
(R_1, R_2)=(I(V_1; Y_1), I(V_2; Y_2)-I(V_1; V_2)),
\end{equation}
assuming that $I(V_1; V_2) < I(V_2; Y_2)$.
Consequently, by Proposition \ref{prop:bin}, we can achieve the whole region
\eqref{eq:marrand} and Theorem~\ref{th:polbinnew} is proved.
Note that if polar coding achieves the rate pair \eqref{eq:rpmar} with
complexity $\Theta(n\log n)$ and a block error probability $O(2^{-n^{\beta}})$,
then for any other rate
pair in the region \eqref{eq:marrand}, there exists a sequence of polar codes
with an increasing block length $n$ whose complexity and block error probability have
the same asymptotic scalings.

{\bf Design of the Scheme.} Set $U_1^{1:n} = V_1^{1:n}G_n$ and $U_2^{1:n} = V_2^{1:n}G_n$.
As in the case of the transmission over a DMC with $V_l$  in place of $X$
and $Y_l$ in place of $Y$ ($l\in \{1, 2\}$), define the sets $\hset_{V_l}$, $\lset_{V_l}$,
$\hset_{V_l\mid Y_l}$, and $\lset_{V_l\mid Y_l}$ for $l\in\{1, 2\}$,
similarly to \eqref{eq:supzero} (except of replacing $U_2$ with $U_l$ and $V$ with $V_l$),
which satisfy
\begin{equation}
\begin{split}  \label{eq:asymptotic limits}
\lim_{n\to \infty}\frac{1}{n} \, | \hset_{V_l}| &= H(V_l),\\
\lim_{n\to \infty}\frac{1}{n} \, | \lset_{V_l}| &= 1 - H(V_l),\\
\lim_{n\to \infty}\frac{1}{n} \, | \hset_{V_l\mid Y_l}| &= H(V_l\mid Y_l),\\
\lim_{n\to \infty}\frac{1}{n} \, | \lset_{V_l\mid Y_l}| &= 1 - H(V_l\mid Y_l).\\
\end{split}
\end{equation}
By thinking of $V_1$ as a side information for $V_2$, we can further define the sets
$\hset_{V_2\mid V_1}$ and $\lset_{V_2\mid V_1}$, which satisfy
\begin{equation} \label{eq:misc}
\begin{split}
\lim_{n\to \infty}\frac{1}{n} \, | \hset_{V_2\mid V_1}|  &= H(V_2\mid V_1),\\
\lim_{n\to \infty}\frac{1}{n} \, | \lset_{V_2\mid V_1}|  &= 1-H(V_2\mid V_1).\\
\end{split}
\end{equation}

First, consider only the point-to-point communication problem between the transmitter
and the first receiver. As discussed in Section~\ref{subsec:general DMC}, for this scenario,
the correct choice is to place the information in those positions of $U_1^{1:n}$ that
are indexed by the set $\iset^{(1)} = \hset_{V_1} \cap \lset_{V_1\mid Y_1}$, which satisfies
\begin{equation} \label{eq:normalized size of I1}
\lim_{n\to\infty}\frac{1}{n} \, | \iset^{(1)}|  = I(V_1;Y_1).\\
\end{equation}
For the point-to-point communication problem between the transmitter and the second receiver,
we know from Section~\ref{subsec:general DMC} that the information has to be placed in those
positions of $U_2^{1:n}$ that are indexed by $\hset_{V_2} \cap \lset_{V_2\mid Y_2}$.

Let us get back to the broadcasting scenario and note that, unlike superposition coding,
for binning the first user does not decode the message intended for the second user.
Consider the following scheme. The first user adopts the point-to-point communication
strategy: it ignores the existence of the second user, and it uses $\iset^{(1)}$ as an
information set. The frozen positions are divided into the two usual subsets
$\fset_{\rm d}^{(1)} = \hset_{V_1}^{\text{c}}$ and
$\fset_{\rm r}^{(1)} = \hset_{V_1} \cap \lset_{V_1\mid Y_1}^{\text{c}}$, which contain positions
s.t., respectively, $U_1^i$  can or cannot be approximately inferred from $U_1^{1:i-1}$.
On the other hand, the second user does not ignore the existence of the first user by putting
his information in $\hset_{V_2} \cap \lset_{V_2\mid Y_2}$. Indeed, $V_1$ and $V_2$ are, in
general, correlated. Hence, the second user puts his information in
$\iset^{(2)}=\hset_{V_2\mid V_1}\cap \lset_{V_2\mid Y_2}$. If $i\in \iset^{(2)}$ then,
since $\iset^{(2)} \subseteq \hset_{V_2\mid V_1}$, the bit $U_2^i$ is approximately independent
of $(U_2^{1:i-1}, V_1^{1:n})$. This implies that $U_2^i$ is suitable to contain information.
Furthermore, since $i\in \lset_{V_2\mid Y_2}$, the bit $U_2^i$ is approximately a deterministic
function of $(U_2^{1:i-1}, Y_2^{1:n})$. This implies that it is also decodable given the channel
output $Y_2^{1:n}$. The remaining positions need to be frozen and can be divided into four
subsets:
\begin{itemize}
\item For $i \in \fset^{(2)}_{\rm r} = \hset_{V_2\mid V_1} \cap \lset_{V_2\mid Y_2}^{\text{c}}$,
$U_2^i$ is chosen uniformly at random, and this randomness is shared between the transmitter
and the second receiver.
\item For $i\in\fset^{(2)}_{\rm d}=\lset_{V_2}$, $U_2^i$ is approximately a deterministic function
of $U_2^{1:i-1}$ and, therefore, its value can be deduced from the past.
\item For $i\in\fset^{(2)}_{\rm out}=\hset_{V_2\mid V_1}^{\text{c}}\cap \lset_{V_2}^{\text{c}} \cap
\lset_{V_2\mid Y_2} $, $U_2^i$ is approximately a deterministic function of
$(U_2^{1:i-1}, V_1^{1:n})$, but it can be deduced also from the channel output $Y_2^{1:n}$.
\item For $i\in \fset^{(2)}_{\rm cr}=\hset_{V_2\mid V_1}^{\text{c}} \cap
\lset_{V_2}^{\text{c}}\cap \lset_{V_2\mid Y_2}^{\text{c}} = \hset_{V_2\mid V_1}^{\text{c}}
\cap \lset_{V_2\mid Y_2}^{\text{c}}$, $U_2^i$ is approximately a deterministic function of
$(U_2^{1:i-1}, V_1^{1:n})$, but it cannot be deduced neither from $U_2^{1:i-1}$ nor from
$Y_2^{1:n}$.
\end{itemize}
The positions belonging to the last set are critical, since, in order to decode them, the receiver
needs to know $V_1^{1:n}$. Indeed, recall that the encoding operation is performed jointly
by the two users, while the first and the second decoder act separately and cannot exchange
any information. The situation is schematically represented in Figure~\ref{fig:marton_sets}.

\begin{figure}[tb]
\begin{center} \includegraphics[width=9.5cm]{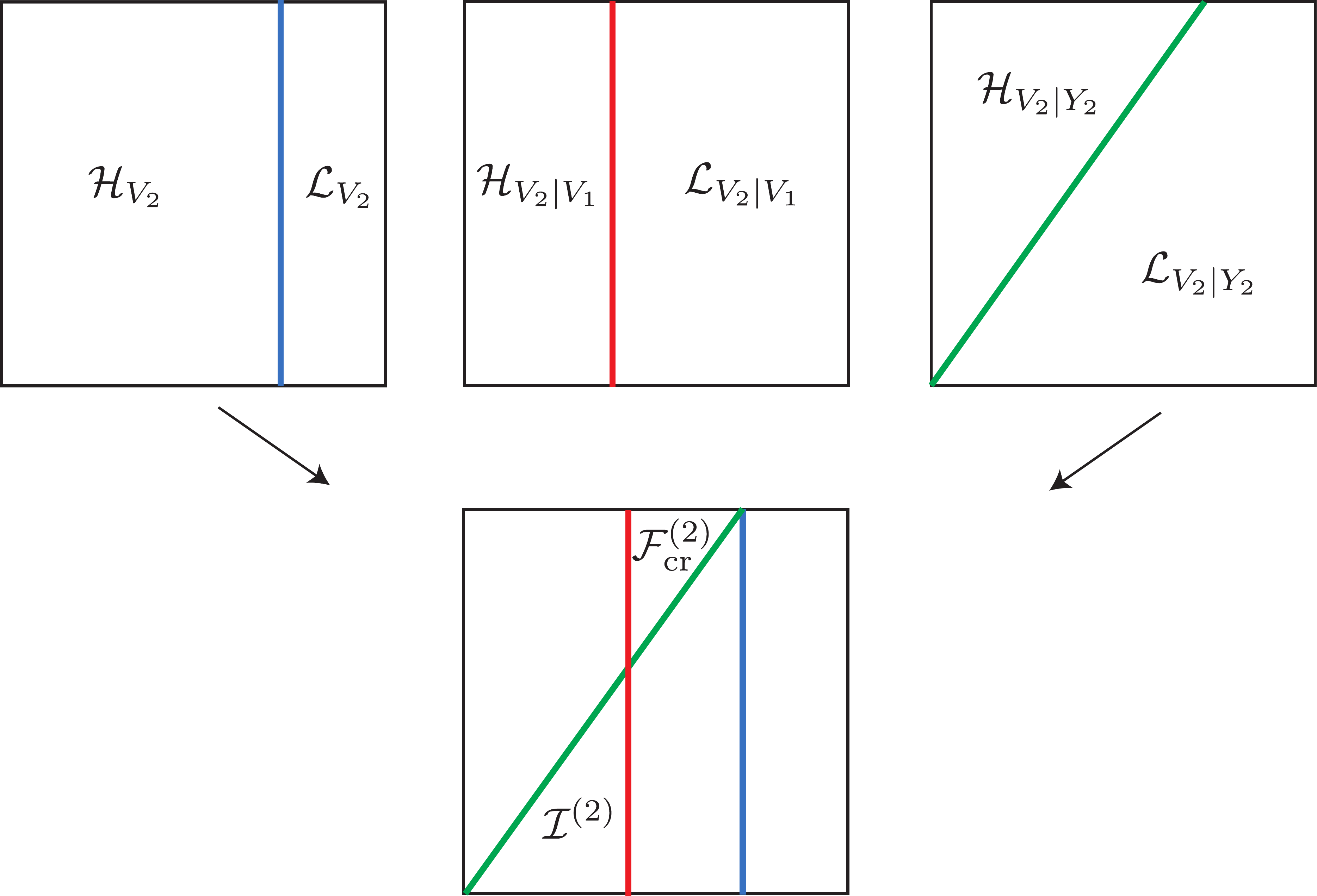}
\end{center}
\caption{Graphical representation of the sets associated to the second user for the
binning scheme: $\iset^{(2)}$ contains the information bits; $\fset^{(2)}_{\rm cr}$ contains
the frozen positions which are critical in the sense that they cannot be inferred neither
from the past $U_2^{1:i-1}$ nor from the channel output $Y_2^{1:n}$.}\label{fig:marton_sets}
\end{figure}

We start by reviewing the \GAG scheme \cite{GAG13ar}. This scheme achieves the rate pair
in \eqref{eq:rpmar}, assuming that the degradedness relation
$p_{Y_2\mid V_2} \succ p_{V_1\mid V_2}$ holds. Note that, under this assumption, we have
$\lset_{V_2\mid V_1} \subseteq \lset_{V_2\mid Y_2}$. Therefore, $\fset^{(2)}_{\rm cr}
\subseteq \lset_{V_2\mid V_1}^{\text{c}}\cap \hset_{V_2\mid V_1}^{\text{c}}$.
Since $| \lset_{V_2\mid V_1}^{\text{c}}\cap \hset_{V_2\mid V_1}^{\text{c}}| = o(n)$,
it is assumed in \cite{GAG13ar} that the bits indexed by
$ \lset_{V_2\mid V_1}^{\text{c}}\cap \hset_{V_2\mid V_1}^{\text{c}}$ are ``genie-given''
from the encoder to the second decoder. The price to be paid for the transmission of these
extra bits is asymptotically negligible. Consequently, the first user places his information
in $\iset^{(1)}$, the second user places his information in $\iset^{(2)}$, and the bits
in the positions belonging to $ \lset_{V_2\mid V_1}^{\text{c}}\cap \hset_{V_2\mid V_1}^{\text{c}}$ are
pre-communicated to the second receiver.

Our goal is to achieve the rate pair \eqref{eq:rpmar} without the degradedness condition
$p_{Y_2\mid V_2} \succ p_{V_1\mid V_2}$. As in the superposition coding scheme, the idea
consists in transmitting $k$ polar blocks and in repeating (``chaining'') some bits from
one block to the following block. To do so, let $\rset$ be a subset of $\iset^{(2)}$ s.t.
$| \rset|  = | \fset^{(2)}_{\rm cr}|$. As usual, it does not matter what subset we pick.
Since the second user cannot reconstruct the bits at the critical positions $\fset^{(2)}_{\rm cr}$,
we use the set $\rset$ to store the critical bits of the previous block. This construction
is schematically represented in Figure~\ref{fig:marton_rep}.

\begin{figure}[tb]
\begin{center} \includegraphics[width=9.5cm]{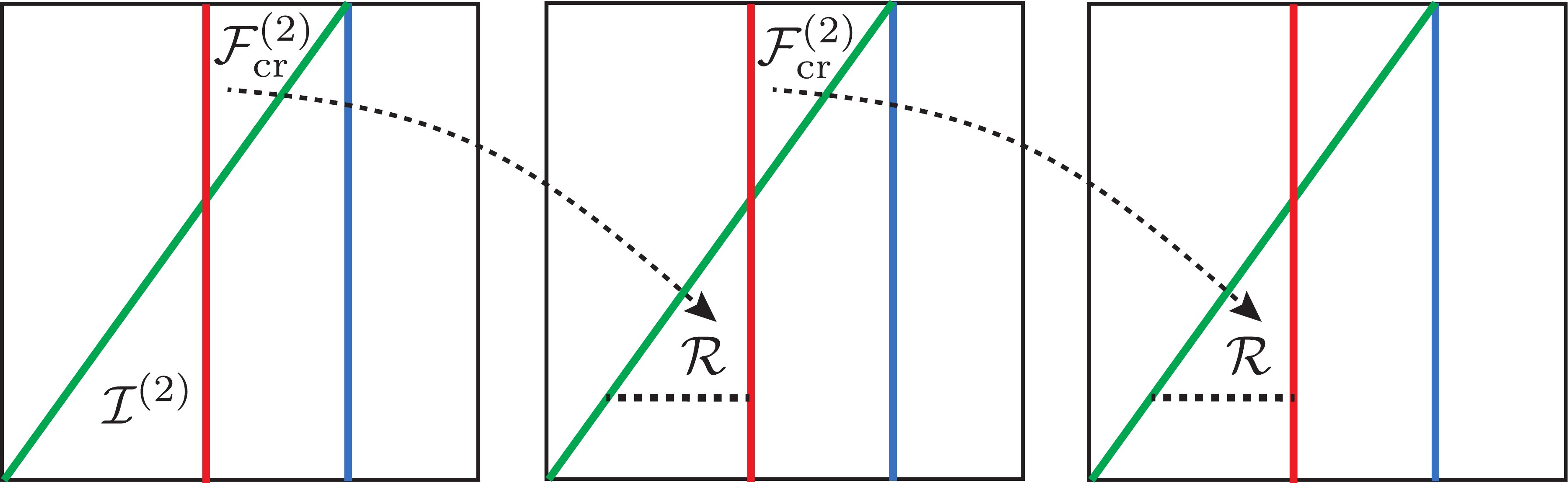}
\end{center}
\caption{Graphical representation of the repetition construction for the binning scheme with $k=3$: the set
$\fset^{(2)}_{\rm cr}$ is repeated into the set $\rset$ of the following block.}\label{fig:marton_rep}
\end{figure}

Let us explain the scheme with some detail. For block $1$, we adopt the point-to-point communication
strategy: the first user puts his information in $\iset^{(1)}$, and the second user in $\iset^{(2)}$.
For block $j$ ($j\in\{2, \cdots, k-1\}$), the first user places again his information in $\iset^{(1)}$.
The second user puts information in the positions indexed by $\iset^{(2)}\setminus \rset$ and repeats
in $\rset$ the bits which were contained in the set $\fset^{(2)}_{\rm cr}$ of block $j-1$. For block
$k$, the second user does not change his strategy, putting information in $\iset^{(2)}\setminus \rset$
and repeating in $\rset$ the bits which were contained in the set $\fset^{(2)}_{\rm cr}$ of block $k-1$.
On the other hand, in the last block, the first user does not convey any information and puts in $\iset^{(1)}$
a fixed sequence which is shared between the encoder and both decoders. Indeed, for block $k$, the positions
indexed by $\fset^{(2)}_{\rm cr}$ are not repeated anywhere. Consequently, the only way in
which the second decoder can reconstruct the bits in $\fset^{(2)}_{\rm cr}$ consists in knowing a priori
the value of $V_1^{1:n}$.

Note that with this scheme, the second user has to decode ``backwards'', starting with block $k$ and ending
with block $1$. In fact, for block $k$, the second user can compute $V_1^{1:n}$ and, therefore, the critical
positions indexed by $\fset^{(2)}_{\rm cr}$ are no longer a problem. Then, for block $j$ ($j \in \{2, \cdots,
k-1\}$), the second user knows the values of the bits in $\fset^{(2)}_{\rm cr}$ from the decoding of the set
$\rset$ of block $j+1$.

Suppose now that the second user wants to decode ``forward'', i.e., starting with block $1$ and
ending with block $k$. Then, the set
$\rset$ is used to store the critical bits of the following block (instead of those ones of the previous
block). In particular, for block $k$, we adopt the point-to-point communication
strategy. For block $j$ ($j\in\{k-1, \cdots, 2\}$), the first user places his information in $\iset^{(1)}$, the second user places his information in the positions indexed by $\iset^{(2)}\setminus \rset$ and repeats
in $\rset$ the bits which were contained in the set $\fset^{(2)}_{\rm cr}$ of block $j+1$. For block
$1$, the second user does not change his strategy, and the first user puts in $\iset^{(1)}$
a shared fixed sequence. Note that in this case the encoding needs to be performed ``backwards''.

{\bf Encoding.} Let us start from the first user.
\newline For block $j$ ($j\in \{1, \cdots, k-1\}$):
\begin{itemize}
\item The information bits are stored in $\{u_1^i\}_{i \in \iset^{(1)}}$.
\item The set $\{u_1^i\}_{i \in \fset_{\rm r}^{(1)}}$ is filled
with a random sequence, which is shared between the transmitter and the first receiver.
\item For $i \in \fset_{\rm d}^{(1)}$, we set
$u_1^i = \arg\max_{u\in \{0, 1\}} {\mathbb P}(U_1^i=u \mid  U_1^{1:i-1}=u_1^{1:i-1})$.
\end{itemize}
For block $k$:
\begin{itemize}
\item The user conveys no information, and $\{u_1^i\}_{i \in \iset^{(1)}}$ contains a fixed sequence
known to the second decoder.
\item The frozen bits are chosen according to the usual rules with the only difference
that the sequence $\{u_1^i\}_{i \in \fset_{\rm r}^{(1)}}$ is shared also with the second decoder.
\end{itemize}
The rate of communication of the first user is given by (see \eqref{eq:normalized size of I1})
\begin{equation}
R_1 = \left(\frac{k-1}{kn}\right) |\iset^{(1)}| = \left(\frac{k-1}{k}\right) I(V_1; Y_1)+o(1),
\end{equation}
where, by choosing a large value of $k$, the rate $R_1$ approaches $I(V_1; Y_1)$. Then, the vector
$v_1^{1:n} = u_1^{1:n} G_n$ is obtained.

Let us now move to the second user.
\newline For block $1$:
\begin{itemize}
\item The information bits are stored in $\{u_2^i\}_{i \in \iset^{(2)}}$.
\item For $i \in \fset^{(2)}_{\rm r}$,
$u_2^i$ is chosen uniformly at random, and its value is supposed to be known to the second decoder.
\item For $i \in \fset^{(2)}_{\rm d}$, $u_2^i$ is set to
$\arg\max_{u\in \{0, 1\}} {\mathbb P}_{U_2^i\mid U_2^{1:i-1}}(u \mid  u_2^{1:i-1})$
\item For $i \in \fset^{(2)}_{\rm out}\cup \fset^{(2)}_{\rm cr}$, $u_2^i$ is set to $\arg\max_{u\in \{0, 1\}}
{\mathbb P}_{U_2^i\mid U_2^{1:i-1}, V_1^{1:n}}(u \mid  u_2^{1:i-1}, v_1^{1:n})$.
\end{itemize}
Observe that the encoder has an access to $v_1^{1:n}$ and, therefore, it can compute
the probabilities above.
\newline For block $j$ ($j\in \{2, \cdots, k\}$):
\begin{itemize}
\item The information bits are placed into $\{u_2^i\}_{i \in \iset^{(2)}\setminus \rset}$.
\item The set $\{u_2^i\}_{i \in \rset}$
contains a copy of the set $\{u_2^i\}_{i \in \fset^{(2)}_{\rm cr}}$ of block $j-1$.
\item The frozen bits are chosen as in block $1$.
\end{itemize}
In order to compute the rate achievable by the second user, first observe that
\begin{equation} \label{eq:rate2}
\begin{split}
\frac{1}{n} \, (|\iset^{(2)}| -| \rset| )
&\stackrel{(\text{a})}{=} \frac{1}{n} \, \Bigl(| \hset_{V_2\mid V_1}\cap
\lset_{V_2\mid Y_2}| -| \hset_{V_2\mid V_1}^{\text{c}}
\cap \lset_{V_2}^{\text{c}}\cap \lset_{V_2\mid Y_2}^{\text{c}}| \Bigr) \\
&\stackrel{(\text{b})}{=} \frac{1}{n} \, \Bigl( | (\hset_{V_2}\cap \lset_{V_2\mid Y_2})
\setminus (\hset_{V_2}\cap \hset_{V_2\mid V_1}^{\text{c}})|
- | (\hset_{V_2}\cap \hset_{V_2\mid V_1}^{\text{c}})
\setminus (\hset_{V_2}\cap \lset_{V_2\mid Y_2})| \Bigr) +o(1)\\
&\stackrel{(\text{c})}{=} \frac{1}{n} \, \Bigl(| \hset_{V_2}\cap \lset_{V_2\mid Y_2}|
- |\hset_{V_2}\cap \hset_{V_2\mid V_1}^{\text{c}}| \Bigr)+o(1)\\
&\stackrel{(\text{d})}{=} \frac{1}{n} \, \Bigl( | \hset_{V_2}\cap \lset_{V_2\mid Y_2}|
- | \hset_{V_2}\cap \lset_{V_2\mid V_1}| \Bigr)+o(1) \\
&\stackrel{(\text{e})}{=} \frac{1}{n} \, \Bigl( | \lset_{V_2 \mid Y_2}\setminus \lset_{V_2}|
- | \lset_{V_2 \mid V_1} \setminus \lset_{V_2}| \Bigr) + o(1) \\
&\stackrel{(\text{f})}{=} I(V_2;Y_2)-I(V_1;V_2)+o(1),
\end{split}
\end{equation}
where equality~(a) holds since $|\rset| = |\fset^{(2)}_{\rm cr}|$, equality~(b)
follows from $\hset_{V_2 \mid V_1} \subseteq \hset_{V_2}$ and
$|[n]\setminus (\hset_{V_2}\cup \lset_{V_2})|=o(n)$, equality~(c) follows from
the identity in \eqref{eq:identity for finite sets} for arbitrary finite sets,
equality~(d) holds since $|[n]\setminus (\hset_{V_2 \mid V_1}\cup \lset_{V_2 \mid V_1})|=o(n)$,
equality~(e) holds since $|[n]\setminus (\hset_{V_2}\cup \lset_{V_2})|=o(n)$,
and equality~(f) follows from the second and fourth equalities in \eqref{eq:asymptotic limits},
as well as from the second equality in \eqref{eq:misc}. Consequently,
\begin{equation}
R_2 = \frac{1}{nk}| \rset|+I(V_2;Y_2)-I(V_1;V_2)+o(1),
\end{equation}
which, as $k$ tends to infinity, approaches the required rate. Then, the vector
$v_2^{1:n} = u_2^{1:n} G_n$ is obtained and, finally, the vector
$x^{1:n} = \phi(v_1^{1:n}, v_2^{1:n})$ is transmitted over the channel.
The encoding complexity per block is $\Theta(n \log n)$.

{\bf Decoding.} Let us start from the first user, which reconstructs $u_1^{1:n}$ from the channel output
$y_1^{1:n}$. For each block, the decision rule is given by
\begin{equation}
\hat{u}_1^i = \left\{ \begin{array}{ll} u_1^i, &
\mbox{if } i \in \fset_{\rm r}^{(1)} \\
\displaystyle\arg\max_{u\in \{0, 1\}} {\mathbb P}_{U_1^i \mid
U_1^{1:i-1}}(u \mid u_1^{1:i-1}),
& \mbox{if } i \in \fset_{\rm d}^{(1)} \\
\displaystyle\arg\max_{u\in \{0, 1\}} {\mathbb P}_{U_1^i \mid
U_1^{1:i-1}, Y_1^{1:n}}(u \mid u_1^{1:i-1}, y_1^{1:n}),
& \mbox{if } i \in \iset^{(1)} \\ \end{array}\right. .
\end{equation}

The second user reconstructs $u_2^{1:n}$ from the channel output $y_2^{1:n}$. As explained before,
the decoding goes ``backwards'', starting from block $k$ and ending with block $1$. For block $k$,
the second decoder knows $v_1^{1:n}$. Hence, the decision rule is given by
\begin{equation} \label{eq:M2}
\hat{u}_2^i = \left\{ \begin{array}{ll}u_2^i &  \mbox{if }
i \in \fset_{\rm r}^{(2)} \\
\displaystyle\arg\max_{u\in \{0, 1\}} {\mathbb P}_{U_2^i \mid
U_2^{1:i-1}}(u \mid u_2^{1:i-1}),& \mbox{if } i \in \fset_{\rm d}^{(2)} \\
\displaystyle\arg\max_{u\in \{0, 1\}} {\mathbb P}_{U_2^i \mid
U_2^{1:i-1}, V_1^{1:n}}(u \mid u_2^{1:i-1}, v_1^{1:n}),&
 \mbox{if } i \in \fset_{\rm out}^{(2)}\cup \fset_{\rm cr}^{(2)} \\
\displaystyle\arg\max_{u\in \{0, 1\}} {\mathbb P}_{U_2^i \mid
U_2^{1:i-1}, Y_2^{1:n}}(u \mid u_2^{1:i-1}, y_2^{1:n}),
& \mbox{if } i \in \iset^{(2)} \\ \end{array}\right. .
\end{equation}
For block $j$ ($j\in\{2, \cdots, k\}$), the decision rule is the same as \eqref{eq:M2} for
$i\not \in \fset_{\rm out}^{(2)}\cup \fset_{\rm cr}^{(2)}$. Indeed,
$\{u_2^i\}_{i \in \fset^{(2)}_{\rm cr}}$ of block $j$ can be deduced from
$\{u_2^i\}_{i \in \rset}$ of block $j+1$, and, for $i \in \fset_{\rm out}^{(2)}$, we set
$\hat{u}_2^i=\arg\max_{u\in \{0, 1\}} {\mathbb P}_{U_2^i \mid  U_2^{1:i-1},
Y_2^{1:n}}(u \mid u_2^{1:i-1}, y_2^{1:n})$. The complexity per
block, under successive cancellation decoding, is $\Theta(n \log n)$.

{\bf Performance.} The block error probability $P_{\rm e}^{(l)}$ for the
$l$-th user ($l\in\{1, 2\}$) can be upper bounded by
\begin{equation}
\begin{split}
P_{\rm e}^{(1)} &\le k \sum_{i \in \iset^{(1)}} Z(U_1^i \mid U_1^{1:i-1},
Y_1^{1:n}) = O(2^{-n^{\beta}}),\\
P_{\rm e}^{(2)} &\le k \sum_{i \in \lset_{V_2\mid Y_2}} Z(U_2^i \mid
U_2^{1:i-1}, Y_2^{1:n}) = O(2^{-n^{\beta}}), \quad \forall\, \beta \in (0, 1/2). \\
\end{split}
\end{equation}


\section{Polar Codes for Marton's Region}\label{sec:combreg}

\subsection{Only Private Messages}\label{subsec:private}

Consider first the case where only private messages are available. The following theorem provides our main result regarding the achievability with polar codes
of Marton's region, which forms the tightest inner
bound known to date for a $2$-user DM-BC without common information (compare with
Theorem~\ref{th:bestreg}).

\begin{theorem}[Polar Codes for Marton's Region] \label{th:polar_Marton_new}
Consider a $2$-user DM-BC $p_{Y_1, Y_2\mid X}$, where $X$ denotes the input
to the channel, taking values on an arbitrary set $\mathcal{X}$,
and $Y_1$, $Y_2$ denote the outputs at the first and second receiver,
respectively. Let $V$, $V_1$, and $V_2$ denote auxiliary binary random variables.
Then, for any joint distribution $p_{V, V_1, V_2}$, for any deterministic
function $\phi: \{0, 1\}^3 \rightarrow \mathcal{X}$
s.t. $X = \phi(V, V_1, V_2)$, and for any rate pair $(R_1, R_2)$ satisfying
the constraints \eqref{eq:best}, there exists a sequence of polar codes with
an increasing block length $n$, which achieves this rate pair with encoding and
decoding complexity $\Theta(n \log n)$ and a block error probability decaying
like $O(2^{-n^{\beta}})$ for any $\beta \in (0, 1/2)$.  \end{theorem}

The proposed coding scheme is a combination of the techniques described in detail
in Sections~\ref{sec:superposition} and \ref{sec:binning}, and it is outlined below.

{\bf Problem Statement.}
Let $(V, V_1, V_2)\sim p_V p_{V_2\mid V}p_{V_1\mid V_2 V}$, and let $X$ be a deterministic
function of $(V, V_1, V_2)$, i.e., $X = \phi(V, V_1, V_2)$. Consider the $2$-user DM-BC
$p_{Y_1, Y_2\mid X}$ s.t. $I(V; Y_1)\le I(V; Y_2)$. The aim is to  achieve the rate pair
\begin{equation}\label{eq:genpair}
(R_1, R_2) = (I(V, V_1; Y_1)-I(V_1; V_2\mid V)-I(V; Y_2), I(V, V_2; Y_2) ).
\end{equation}
 Once we have accomplished this, we will see that a slight modification of this scheme
 allows us to achieve, in addition, the rate pair
 \begin{equation}\label{eq:genpair2}
 (R_1, R_2) = (I(V, V_1; Y_1), I(V_2; Y_2\mid V)-I(V_1; V_2\mid V) ).
 \end{equation}
Therefore, by Proposition \ref{prop:combreg}, we can achieve the whole rate region in
\eqref{eq:best} by using polar codes.
Note that if polar coding achieves the rate pairs \eqref{eq:genpair} and
\eqref{eq:genpair2} with complexity $\Theta(n\log n)$ and a block error probability
$O(2^{-n^{\beta}})$,
then  for any other rate
pair in the region \eqref{eq:best}, there exists a sequence of polar codes
with an increasing block length $n$ whose complexity and block error probability have
the same asymptotic scalings.

{\bf Sketch of the Scheme.}
Set $U_0^{1:n} = V^{1:n} G_n$, $U_1^{1:n} = V_1^{1:n} G_n$, and $U_2^{1:n} = V_2^{1:n} G_n$.
Then, the idea is that $U_1^{1:n}$ carries the message of the first user, while $U_0^{1:n}$
and $U_2^{1:n}$ carry the message of the second user. The second user will be able to decode
only his message, namely, $U_0^{1:n}$ and $U_2^{1:n}$. On the other hand, the first user will
decode both his message, namely, $U_1^{1:n}$, and a part of the message of the second user,
namely, $U_0^{1:n}$. In a nutshell, the random variable $V$ comes from the superposition coding
scheme, because $U_0^{1:n}$ is decodable by both users, but carries information meant only for
one of them. The random variables $V_1$ and $V_2$ come from the binning scheme, since the first
user decodes $U_1^{1:n}$ and the second user decodes $U_2^{1:n}$, i.e., each user decodes only
his own information.

Let the sets $\hset_{V}$, $\lset_{V}$,
$\hset_{V \mid Y_l}$, and $\lset_{V \mid Y_l}$ for $l\in \{1, 2\}$ be defined as in \eqref{eq:supzero},
where these subsets of $[n]$ satisfy \eqref{eq:supfirst}. In analogy to Sections~\ref{sec:superposition}
and \ref{sec:binning} let us also define the following sets $(l \in \{1, 2\})$:
\begin{equation} \label{eq:martonzero}
\vspace*{-0.2cm}
\begin{split}
\hset_{V_l \mid V} &= \{i \in [n] \colon Z(U_l^i \mid  U_l^{1:i-1}, U_0^{1:n}) \ge 1- \delta_n\}, \\
\lset_{V_l \mid V} &= \{i \in [n] \colon Z(U_l^i \mid  U_l^{1:i-1}, U_0^{1:n}) \le \delta_n\}, \\
\hset_{V_l \mid V, Y_l} &= \{i \in [n] \colon Z(U_l^i \mid  U_l^{1:i-1}, U_0^{1:n}, Y_l^{1:n}) \ge 1- \delta_n\},  \\
\lset_{V_l \mid V, Y_l} &= \{i \in [n] \colon Z(U_l^i \mid  U_l^{1:i-1}, U_0^{1:n}, Y_l^{1:n}) \le \delta_n\}, \\
\hset_{V_1 \mid V, V_2} &= \{i \in [n] \colon Z(U_1^i \mid  U_1^{1:i-1}, U_0^{1:n}, U_2^{1:n}) \ge 1- \delta_n\},  \\
\lset_{V_1 \mid V, V_2} &= \{i \in [n] \colon Z(U_1^i \mid  U_1^{1:i-1}, U_0^{1:n}, U_2^{1:n}) \le \delta_n\},
\end{split}
\end{equation}
which satisfy
\begin{eqnarray}  \label{eq:asymptotic limits - marton}
\begin{split}
&\lim_{n\to \infty}\frac{1}{n} \, | \hset_{V_l \mid V}| = H(V_l \mid V), \qquad \qquad
\lim_{n\to \infty}\frac{1}{n} \, | \lset_{V_l \mid V}| = 1 - H(V_l \mid V),\\
&\lim_{n\to \infty}\frac{1}{n} \, | \hset_{V_l \mid V, Y_l}| = H(V_l \mid V, Y_l), \qquad
\lim_{n\to \infty}\frac{1}{n} \, | \lset_{V_l \mid V, Y_l}| = 1 - H(V_l \mid V, Y_l),\\
&\lim_{n\to \infty}\frac{1}{n} \, | \hset_{V_1 \mid V, V_2}| = H(V_1 \mid V, V_2), \quad \;
\lim_{n\to \infty}\frac{1}{n} \, | \lset_{V_1 \mid V, V_2}| = 1 - H(V_1 \mid V, V_2).
\end{split}
\end{eqnarray}

First, consider the subsets of positions of $U_0^{1:n}$.  The set $\iset^{(2)}_{\rm sup} = \hset_V
\cap \lset_{V\mid Y_2}$ contains the positions which are decodable by the second user, and the
set $\iset^{(1)}_{v} = \hset_V \cap \lset_{V\mid Y_1}$ contains the positions which are decodable
by the first user. Recall that $U_0^{1:n}$ needs to be decoded by both users, but contains
information only for the second user.

Second, consider the subsets of positions of $U_2^{1:n}$.
The set $\iset^{(2)}_{\rm bin} = \hset_{V_2\mid V} \cap \lset_{V_2\mid V, Y_2}$ contains the positions
which are decodable by the second user. Recall that $U_2^{1:n}$ needs to be decoded only by the second
user, and it contains part of his message.

Third, consider the subsets of positions of $U_1^{1:n}$.
The set $\iset^{(1)} = \hset_{V_1\mid V} \cap \lset_{V\mid Y_2}$ contains the
positions which are decodable by the first user. Recall that $U_1^{1:n}$ needs to be decoded by the
first user, and it contains only his message. However, the first user cannot decode $U_2^{1:n}$ and,
therefore, this user cannot infer $V_2^{1:n}$. Consequently, the positions in the set $\fset_{\rm cr}^{(1)} =
\hset_{V_1\mid V, V_2}^c\cap \lset_{V_1\mid V}^c \cap \lset_{V_1\mid V, Y_1}^c$ are critical. Indeed,
for $i\in \fset_{\rm cr}^{(1)}$, the bit $U_1^i$ is approximately a deterministic function of
$(U_1^{1:i-1}, U_0^{1:n}, U_2^{1:n})$, but it cannot be deduced from $(U_1^{1:i-1}, U_0^{1:n}, Y_1^{1:n})$.

In order to achieve the rate pair \eqref{eq:genpair}, $k$ polar blocks are transmitted, and three
different ``chaining'' constructions are used. The first and the second chaining come from superposition
coding, and the last one comes from binning.

\begin{figure}[p]
\centering
\subfigure[Subsets of $U_0^{1:n}$.]
{\includegraphics[width=9.5cm]{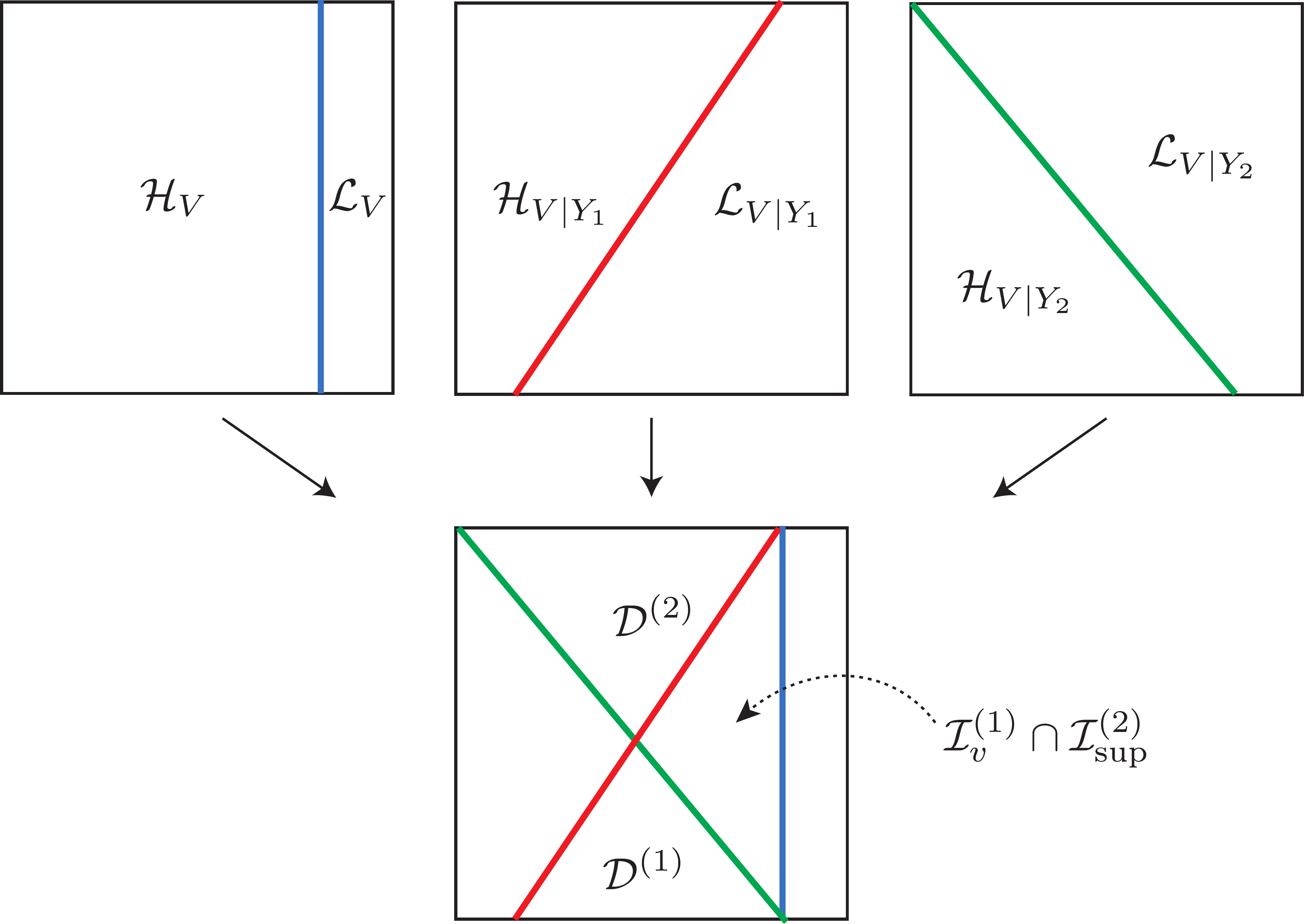}}\\
\subfigure[Subsets of $U_2^{1:n}$.]
{\includegraphics[width=6.5cm]{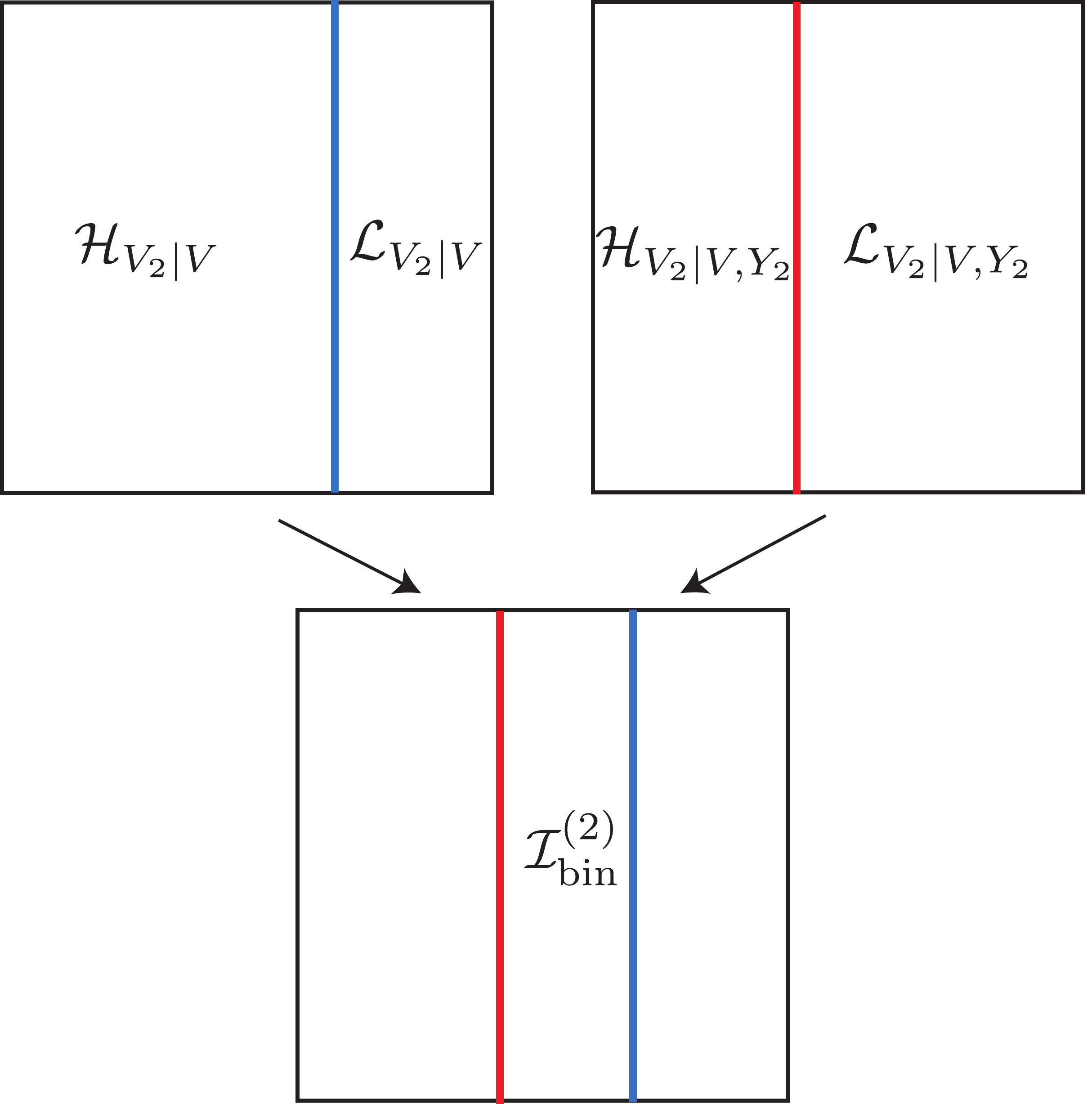}}\\
\subfigure[Subsets of $U_1^{1:n}$.]
{\includegraphics[width=9.5cm]{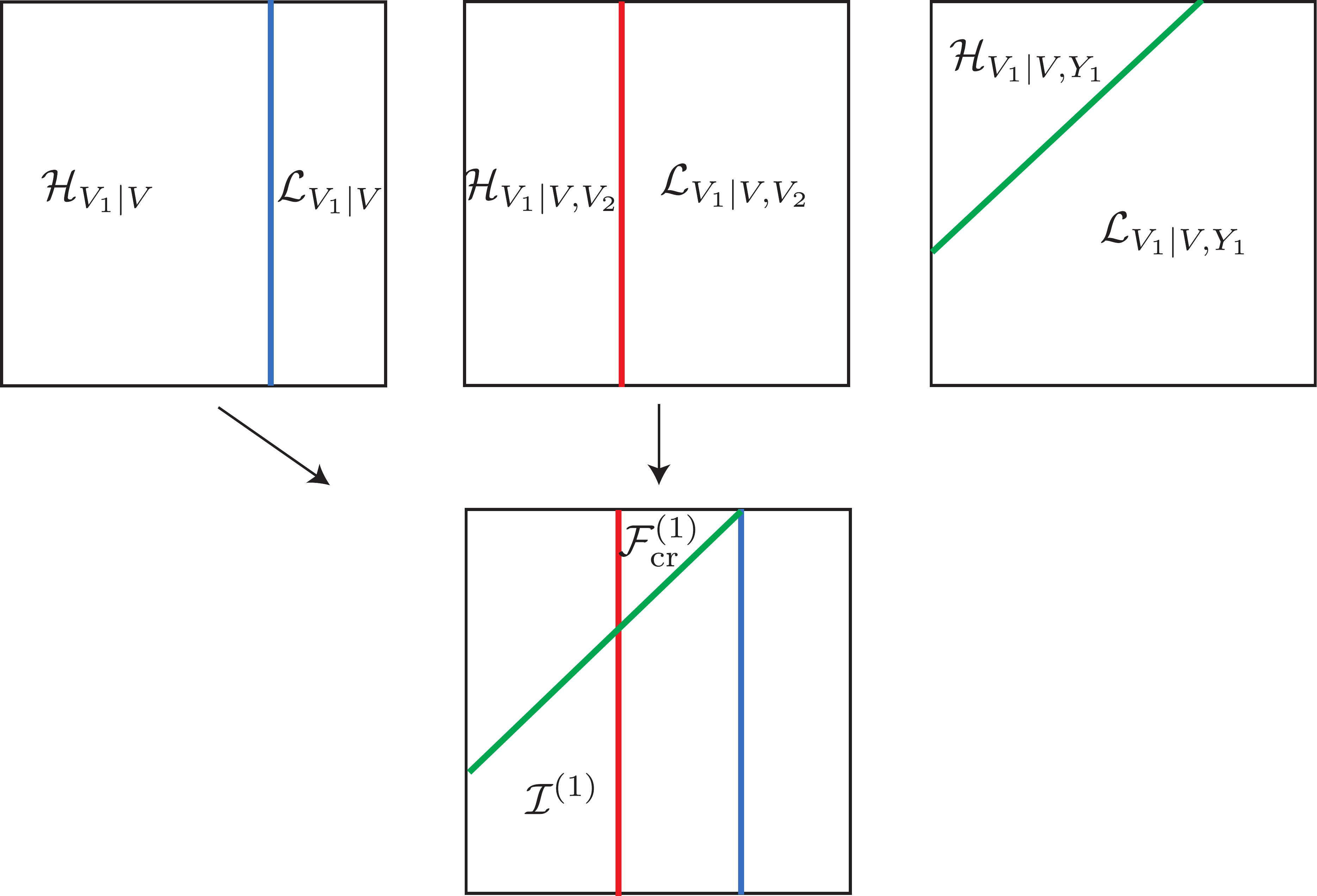}}\\
\caption{Graphical representation of the sets associated to the three auxiliary random variables in the scheme which achieves Marton's region with only private messages \eqref{eq:best}.}\label{fig:gen1}
\end{figure}

First, define $\dset^{(2)} = \iset^{(2)}_{\rm sup}\setminus \iset^{(1)}_v$ and $\dset^{(1)} = \iset^{(1)}_v
\setminus \iset^{(2)}_{\rm sup}$, as in \eqref{eq:dset^2} and \eqref{eq:dset^1}, respectively. The former
set contains the positions of $U_0^{1:n}$ which are decodable by the second user but not by the first, while
the latter contains the positions of $U_0^{1:n}$ which are decodable by the first user but not by the second.
Let $\rset_{\rm sup}$ be a subset of $\dset^{(2)}$ s.t. $|\rset_{\rm sup}| = |\dset^{(1)}|$. In block $1$, fill
$\dset^{(1)}$ with information for the second user, and set the bits indexed by $\dset^{(2)}$ to a fixed known
sequence. In block $j$ ($j \in \{2, \cdots, k-1\}$), fill $\dset^{(1)}$ again with information for the second
user, and repeat the bits which were contained in the set $\dset^{(1)}$ of block $j-1$ into the
positions indexed by $\rset_{\rm sup}$ of block $j$. In the final block $k$,
put a known sequence in the positions indexed by $\dset^{(1)}$, and
repeat in the positions indexed by $\rset_{\rm sup}$ the bits which were contained
in the set $\dset^{(1)}$ of block $k-1$. In all the blocks, fill $\iset^{(1)}_v\cap \iset^{(2)}_{\rm sup}$ with
information for the second user. In this way, both users will be able to decode a fraction of the bits
of $U_0^{1:n}$ that is roughly equal to $I(V; Y_1)$. The bits in these positions contain information
for the second user.

Second, define $\bset^{(2)} = \dset^{(2)}\setminus \rset_{\rm sup}$, and let $\bset^{(1)}$ be a subset of
$\iset^{(1)}$ s.t. $|\bset^{(1)}|=|\bset^{(2)}|$. Note that $\bset^{(2)}$ contains positions of $U_0^{1:n}$,
and $\bset^{(1)}$ contains positions of $U_1^{1:n}$. For block $j$ ($j\in \{2, \cdots, k\}$), we fill
$\bset^{(2)}$ with information for the second user, and we repeat these bits into the positions indexed by
$\bset^{(1)}$ of block $j-1$. In this way, both users will be able to decode a fraction of the bits
of $U_0^{1:n}$ that is roughly equal to $I(V; Y_2)$ (recall that $I(V; Y_1)\le I(V; Y_2)$).
Again, the bits in these positions contain information for the second user.

Third, let $\rset_{\rm bin}$ be a subset of $\iset^{(1)}$ s.t. $|\rset_{\rm bin}|=|\fset_{\rm cr}^{(1)}|$.
Since the first user cannot reconstruct the bits at the critical positions $\fset_{\rm cr}^{(1)}$, we use
the set $\rset_{\rm bin}$ to store the critical bits of the following block. For block $k$, the first user
places all his information in $\iset^{(1)}$. For block $j$ ($j\in \{1, \cdots, k-1\}$), the first user
places all his information in $\iset^{(1)}\setminus (\rset_{\rm bin}\cup \bset^{(1)})$, repeats in
$\rset_{\rm bin}$ the bits in $\fset_{\rm cr}^{(1)}$ for block $j+1$, and repeats in $\bset^{(1)}$ the bits
in $\bset^{(2)}$ for block $j+1$. The second user puts part of his information in $\iset^{(2)}_{\rm bin}$
(which is a subset of the positions of $U^{1:n}_2$) for all the blocks except for the first, in which
$\iset^{(2)}_{\rm bin}$ contains a fixed sequence which is shared between the encoder and both decoders.
Indeed, for block~1, the positions indexed by $\fset_{\rm cr}^{(1)}$ are not repeated anywhere, and the
only way in which the second decoder can reconstruct those bits consists in knowing a-priori the value of
$V_2^{1:n}$. The situation is schematically represented in Figures~\ref{fig:gen1} and \ref{fig:gen2}.

\begin{figure}[t]
\centering
\includegraphics[width=11cm]{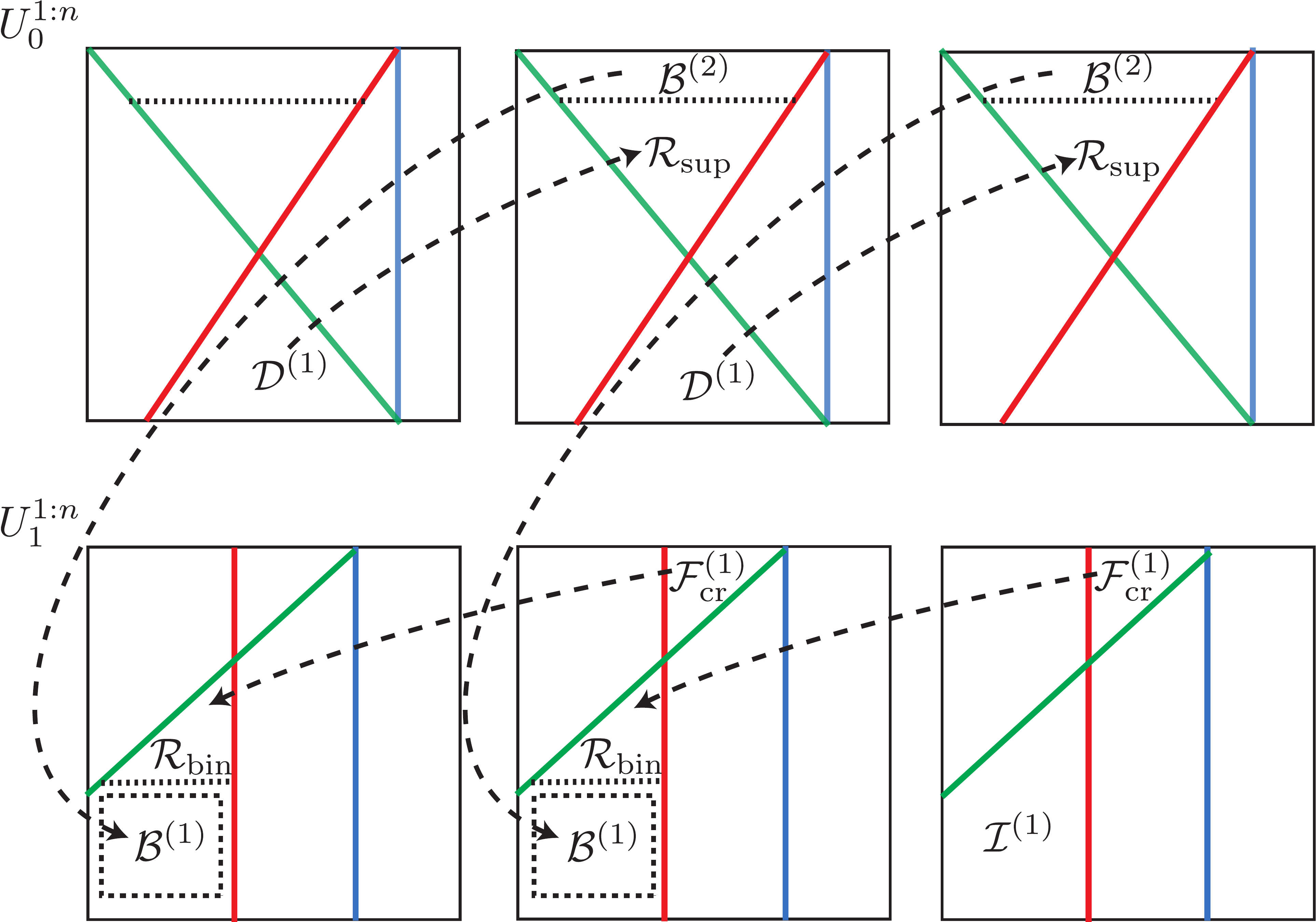}
\caption{Graphical representation of the repetition constructions for Marton's region with $k=3$: the set $\dset^{(1)}$ is repeated into the set $\rset_{\rm sup}$ of the following block;
the set $\bset^{(2)}$ is repeated into the set $\bset^{(1)}$ of the previous block; the set $\fset^{(1)}_{\rm cr}$
is repeated into the set $\rset_{\rm bin}$ of the previous block.}\label{fig:gen2}
\end{figure}

The encoding of $U_0^{1:n}$ is performed ``forward'', i.e., from block $1$ to block $k$; the encoding of $U_1^{1:n}$
is performed ``backwards'', i.e., from block $k$ to block $1$; the encoding of $U_2^{1:n}$ can be performed in any
order. The first user decodes $U_0^{1:n}$ and $U_1^{1:n}$ ``forward''; the second user decodes $U_0^{1:n}$
``backwards'' and can decode $U_2^{1:n}$ in any order.

With this polar coding scheme, by letting $k$ tend to infinity, the first user decodes a fraction of the positions of $U_1^{1:n}$ containing his own message, which is given by
\begin{equation}
\begin{split}
R_1 &= \frac{1}{n}(|\iset^{(1)}| - |\bset^{(1)}| - |\rset_{\rm bin}|) =  I(V_1; Y_1\mid V) - I(V_1; V_2\mid V) - (I(V; Y_2)-I(V; Y_1)) \\
&= I(V, V_1; Y_1) - I(V_1; V_2\mid V) - I(V; Y_2).
\end{split}
\end{equation}
The information for the second user is spread between the positions of $U_0^{1:n}$ and the positions of $U_2^{1:n}$
for a total rate, which, as $k$ tends to infinity, is given by
\begin{equation}
R_2 = \frac{1}{n}(|\iset^{(2)}_{\rm sup}|+|\iset^{(2)}_{\rm bin}|)= I(V; Y_2) + I(V_2; Y_2 \mid V) = I(V, V_2; Y_2).
\end{equation}

It is possible to achieve the rate pair \eqref{eq:genpair2} with a scheme similar to the one described above by
swapping the roles of the two users. Since $I(V; Y_1)\le I(V; Y_2)$, only the first and the third chaining constructions are required. Indeed, the set which has the role of $\bset^{(2)}$ is empty in this scenario.

As our schemes consist in the repetition of polar blocks, the encoding and decoding complexity per block is
$\Theta (n \log n )$, and the block error probability decays like $O(2^{-n^{\beta}})$ for any $\beta \in (0, 1/2)$.

\subsection{Private and Common Messages: MGP Region} \label{subsec: polar_MGP_new}

Finally, consider the case of a $2$-user DM-BC with both common and private messages.
Our most general result consists in the construction of polar codes which achieve
the MGP region \eqref{eq:bestcom}.

\begin{theorem}[Polar Codes for MGP Region]
\label{th:polar_MGP_new}
Consider a $2$-user DM-BC $p_{Y_1, Y_2\mid X}$, where $X$ denotes the input to the channel,
taking values on an arbitrary set $\mathcal{X}$, and $Y_1$, $Y_2$ denote the outputs at
the first and second receiver, respectively. Let $R_0$, $R_1$, and $R_2$ designate the rates of
the common message and the two private messages of the two users, respectively.
Let $V$, $V_1$, and $V_2$ denote auxiliary binary random variables. Then, for any joint
distribution $p_{V, V_1, V_2}$, for any deterministic function $\phi \colon \{0, 1\}^3 \rightarrow \mathcal{X}$
s.t. $X = \phi(V, V_1, V_2)$, and for any rate triple $(R_0, R_1, R_2)$ satisfying
the constraints \eqref{eq:bestcom}, there exists a sequence of polar codes with
an increasing block length $n$ which achieves this rate triple with encoding and
decoding complexity $\Theta(n \log n)$ and a block error probability decaying
like $O(2^{-n^{\beta}})$ for any $\beta \in (0, 1/2)$.  \end{theorem}

The polar coding scheme follows the ideas outlined in Section \ref{subsec:private}. Recall
that $U_0^{1:n}$ is decoded by both users. Then, we put the common information
in the positions of $U_0^{1:n}$ which previously contained private information meant
only for one of the users. The common rate is clearly upper bounded by
$\min\{I(V; Y_1), I(V; Y_2)\}$. The remaining four inequalities of \eqref{eq:bestcom}
are equivalent to the conditions in  \eqref{eq:best} with the only difference that a
portion $R_0$ of the private information for one of the users has been converted into
common information. This suffices to achieve the required rate region.

\section{Conclusions}\label{sec:conclusion}

Extending the work by Goela, Abbe, and Gastpar \cite{GAG13ar}, we have shown how to
construct polar codes for the $2$-user discrete memoryless broadcast channel (DM-BC)
that achieve the superposition and binning regions. By combining these two strategies,
we achieve any rate pair inside Marton's region with both common and private messages.
This rate region is tight for all classes of broadcast channels with known capacity
regions and it is also known as the Marton-Gelfand-Pinsker (MGP) region.
The described coding techniques possess the usual advantages of polar codes,
i.e., encoding and decoding complexity of $\Theta (n \log n)$ and block error probability
decaying like $O(2^{-n^{\beta}})$ for any $\beta\in (0, 1/2)$, and they can be easily
extended to obtain inner bounds for the
$K$-user DM-BC in a low-complexity fashion.

We conclude by remarking that the chaining constructions used to align the polarized indices
do not rely on the specific structure of the broadcast channel. Indeed, similar techniques
have been considered, independently of this work, in the context of interference networks
\cite{WaS14} and, in general, we believe that they can be adapted to the design of polar
coding schemes for a variety of multi-user scenarios.

\section*{Acknowledgment}
We wish to thank the Associate Editor, Henry Pfister, for efficiently handling our manuscript.
The work of M.~Mondelli, S.~H.~Hassani and~R.~Urbanke was supported by grant No.
200020\_146832/1 of the Swiss National Science Foundation. The work of I. Sason
was supported by the Israeli Science Foundation (ISF), grant number 12/12.

\bibliographystyle{IEEEtran}
\bibliography{lth,lthpub}

\newcommand{\SortNoop}[1]{}
\begin{thebibliography}{10}
\providecommand{\url}[1]{#1}
\csname url@samestyle\endcsname
\providecommand{\newblock}{\relax}
\providecommand{\bibinfo}[2]{#2}
\providecommand{\BIBentrySTDinterwordspacing}{\spaceskip=0pt\relax}
\providecommand{\BIBentryALTinterwordstretchfactor}{4}
\providecommand{\BIBentryALTinterwordspacing}{\spaceskip=\fontdimen2\font plus
\BIBentryALTinterwordstretchfactor\fontdimen3\font minus
  \fontdimen4\font\relax}
\providecommand{\BIBforeignlanguage}[2]{{%
\expandafter\ifx\csname l@#1\endcsname\relax
\typeout{** WARNING: IEEEtran.bst: No hyphenation pattern has been}%
\typeout{** loaded for the language `#1'. Using the pattern for}%
\typeout{** the default language instead.}%
\else
\language=\csname l@#1\endcsname
\fi
#2}}
\providecommand{\BIBdecl}{\relax}
\BIBdecl

\bibitem{Ari09}
E.~{Ar\i kan}, ``Channel polarization: A method for constructing
  capacity-achieving codes for symmetric binary-input memoryless channels,''
  \emph{IEEE Trans. on Information Theory}, vol.~55, no.~7, pp. 3051--3073,
  July 2009.

\bibitem{ArT09}
E.~{Ar\i kan} and E.~{Telatar}, ``{On the rate of channel polarization},'' in
  \emph{Proc. of the IEEE International Symposium on Information Theory},
  Seoul, South Korea, July 2009, pp. 1493--1495.

\bibitem{HMTU13}
S.~H. Hassani, R.~Mori, T.~Tanaka, and R.~Urbanke, ``Rate-dependent analysis of
  the asymptotic behavior of channel polarization,'' \emph{IEEE Trans. on
  Information Theory}, vol.~59, no.~4, pp. 2267--2276, April 2013.

\bibitem{Ar10}
E.~{Ar\i kan}, ``{Source polarization},'' in \emph{Proc. of the IEEE
  International Symposium on Information Theory}, Austin, Texas, June 2010, pp.
  899--903.

\bibitem{KoU09}
S.~B. Korada and R.~Urbanke, ``Polar codes are optimal for lossy source
  coding,'' \emph{IEEE Trans. on Information Theory}, vol.~56, no.~4, pp.
  1751--1768, April 2010.

\bibitem{Kor09thesis}
S.~B. Korada, ``Polar codes for channel and source coding,'' Ph.D.
  dissertation, EPFL, Lausanne, Switzerland, July 2009.

\bibitem{Ar12}
E.~Arikan, ``Polar coding for the {S}lepian-{W}olf problem based on monotone
  chain rules,'' in \emph{Proc. of the IEEE International Symposium on
  Information Theory}, MIT, Cambridge, USA, July 2012, pp. 571--575.

\bibitem{AbT12}
E.~Abbe and I.~E. Telatar, ``Polar codes for the $m$-user multiple access
  channel,'' \emph{IEEE Trans. on Information Theory}, vol.~58, no.~8, pp.
  5437--5448, August 2012.

\bibitem{TSV12}
I.~Tal, A.~Sharov, and A.Vardy, ``Constructing polar codes for non-binary
  alphabets and {MAC}s,'' in \emph{Proc. of the IEEE International Symposium on
  Information Theory}, Cambridge, MA, July 2012, pp. 2142--2146.

\bibitem{STYe10}
E.~{\c Sa\c so\u glu}, I.~E. Telatar, and E.~Yeh, ``Polar codes for the
  two-user binary-input multiple-access channel,'' in \emph{Proc. of the IEEE
  Information Theory Workshop}, Cairo, Egypt, January 2010, pp. 1--5.

\bibitem{MKLK13}
H.~Mahdavifar, M.~El-Khamy, J.~Lee, and I.~Kang, ``Achieving the uniform rate
  region of multiple access channels using polar codes,'' July 2013, [Online].
  Available: http://arxiv.org/pdf/1307.2889v1.pdf.

\bibitem{NT13}
R.~Nasser and E.~Telatar, ``Polar codes for arbitrary {DMC}s and arbitrary
  {MAC}s,'' November 2013, http://arxiv.org/pdf/1311.3123v1.pdf.

\bibitem{GAG12}
N.~Goela, E.~Abbe, and M.~Gastpar, ``Polar codes for the deterministic
  broadcast channel,'' in \emph{Proc. of the International Zurich Seminar on
  Communications}, February 2012, pp. 51--54.

\bibitem{GAG13}
------, ``Polar codes for broadcast channels,'' in \emph{Proc. of the IEEE
  International Symposium on Information Theory}, Istanbul, Turkey, July 2013,
  pp. 1127--1131.

\bibitem{GAG13ar}
------, ``Polar codes for broadcast channels,'' January 2013, [Online].
  Available: http://arxiv.org/pdf/1301.6150v1.pdf.

\bibitem{AKV11}
K.~Appaiah, O.~Koyluoglu, and S.~Vishwanath, ``Polar alignment for interference
  networks,'' in \emph{Proc. of the Allerton Conference on Communication,
  Control, and Computing}, Monticello, Illinois, September 2011, pp. 240--246.

\bibitem{WaS14}
L.~Wang and E.~{\c Sa\c so\u glu}, ``Polar coding for interference networks,''
  in \emph{Proc. of the IEEE International Symposium on Information Theory},
  Honolulu, Hawaii, USA, July 2014, pp. 311--315.

\bibitem{MK12}
M.~Karzand, ``Polar codes for degraded relay channels,'' in \emph{Proc. Intern.
  Zurich Seminar on Comm.}, February 2012, pp. 59--62.

\bibitem{ARTKS10}
M.~Andersson, V.~Rathi, R.~Thobaben, J.~Kliewer, and M.~Skoglund, ``Nested
  polar codes for wiretap and relay channels,'' \emph{IEEE Communications
  Letters}, vol.~14, no.~8, pp. 752--754, August 2010.

\bibitem{MV11}
H.~Mahdavifar and A.~Vardy, ``Achieving the secrecy capacity of wiretap
  channels using polar codes,'' \emph{IEEE Trans. on Information Theory},
  vol.~57, no.~10, pp. 6428--6443, October 2011.

\bibitem{KG11}
O.~O. Koyluoglu and H.~E. Gamal, ``Polar coding for secure transmission and key
  agreement,'' \emph{IEEE Trans. on Information Forensics Security}, vol.~7,
  no.~5, pp. 1472--1483, October 2012.

\bibitem{HoS10}
E.~Hof and S.~Shamai, ``Secrecy-achieving polar-coding,'' in \emph{Proc. of the
  IEEE Information Theory Workshop}, Dublin, Ireland, September 2010, pp. 1--5.

\bibitem{SaV13}
E.~{\c Sa\c so\u glu} and A.~Vardy, ``A new polar coding scheme for strong
  security on wiretap channels,'' in \emph{Proc. of the IEEE International
  Symposium on Information Theory}, Istanbul, Turkey, July 2013, pp.
  1117--1121.

\bibitem{AS13}
M.~Andersson, R.~Schaefer, T.~Oechtering, and M.~Skoglund, ``Polar coding for
  bidirectional broadcast channels with common and confidential messages,''
  \emph{IEEE Journal on Selected Areas in Communications}, vol.~31, no.~9, pp.
  1901--1908, September 2013.

\bibitem{BurshteinS13}
D.~Burshtein and A.~Strugatski, ``Polar write once memory codes,'' \emph{IEEE
  Trans. on Information Theory}, vol.~59, no.~8, pp. 5088--5101, August 2013.

\bibitem{HSST13}
E.~Hof, I.~Sason, S.~Shamai, and C.~Tian, ``Capacity-achieving polar codes for
  arbitrarily-permuted parallel channels,'' \emph{IEEE Trans. on Information
  Theory}, vol.~59, no.~3, pp. 1505--1516, March 2013.

\bibitem{SaP14}
A.~G. Sahebi and S.~S. Pradhan, ``Polar codes for multi-terminal
  communications,'' in \emph{Proc. of the IEEE International Symposium on
  Information Theory}, Honolulu, Hawaii, USA, July 2014, pp. 316--320.

\bibitem{Ber73}
P.~P. Bergmans, ``Random coding theorem for broadcast channels with degraded
  components,'' \emph{IEEE Trans. on Information Theory}, vol.~19, no.~2, pp.
  197--207, March 1973.

\bibitem{M79}
K.~Marton, ``A coding theorem for the discrete memoryless broadcast channel,''
  \emph{IEEE Trans. on Information Theory}, vol.~25, no.~3, pp. 306--311, May
  1979.

\bibitem{Co72}
T.~M. Cover, ``Broadcast channels,'' \emph{IEEE Trans. on Information Theory},
  vol.~18, no.~1, pp. 2--14, Jan. 1972.

\bibitem{WSBK13}
L.~Wang, E.~{\c Sa\c so\u glu}, B.~Bandemer, and Y.-H. Kim, ``A comparison of
  superposition coding schemes,'' in \emph{Proc. of the IEEE International
  Symposium on Information Theory}, Istanbul, Turkey, July 2013, pp. 2970 --
  2974.

\bibitem{GelPin80}
S.~I. Gelfand and M.~S. Pinsker, ``Capacity of a broadcast channel with one
  deterministic component,'' \emph{Probl. Peredachi Inf.}, vol.~16, no.~1, pp.
  24--34, 1980.

\bibitem{Liang05thesis}
Y.~Liang, ``Multiuser communications with relaying and user cooperation,''
  Ph.D. dissertation, University of Illinois, Urbana-Champaign, Illinois, USA,
  2005.

\bibitem{LiaKra07}
Y.~Liang and G.~Kramer, ``Rate regions for relay broadcast channels,''
  \emph{IEEE Trans. on Information Theory}, vol.~53, no.~10, pp. 3517--3535,
  October 2007.

\bibitem{LiKP11}
Y.~Liang, G.~Kramer, and H.~V. Poor, ``On the equivalence of two achievable
  regions for the broadcast channel,'' \emph{IEEE Trans. on Information
  Theory}, vol.~57, no.~1, pp. 95 -- 100, January 2011.

\bibitem{kim:nit}
A.~E. Gamal and Y.-H. Kim, \emph{Network Information Theory}.\hskip 1em plus
  0.5em minus 0.4em\relax Cambridge University Press, 2011.

\bibitem{Kr07}
G.~Kramer, ``Topics in multi-user information theory,'' \emph{Foundations and
  Trends in Communications and Information Theory}, vol.~4, no. 4-5, pp.
  265--444, April 2007.

\bibitem{GGNY14}
Y.~Geng, A.~Gohari, C.~Nair, and Y.~Yu, ``On {M}arton's inner bound and its
  optimality for classes of product broadcast channels,'' \emph{IEEE Trans. on
  Information Theory}, vol.~60, no.~1, pp. 22--41, January 2014.

\bibitem{HKU09}
S.~H. Hassani, S.~B. Korada, and R.~Urbanke, ``The compound capacity of polar
  codes,'' in \emph{47th Annual Allerton Conference on Communication, Control,
  and Computing}, October 2009, pp. 16 -- 21.

\bibitem{HRunipol}
S.~H. Hassani and R.~Urbanke, ``Universal polar codes,'' Dec. 2013, [Online].
  Available: http://arxiv.org/pdf/1307.7223v2.pdf.

\bibitem{SaW13}
E.~\c{S}a\c{s}o\u{g}lu and L.~Wang, ``Universal polarization,'' in \emph{Proc.
  of the IEEE International Symposium on Information Theory}, Honolulu, Hawaii,
  USA, July 2014, pp. 1456--1460.

\bibitem{SAT09}
E.~{\c Sa\c so\u glu}, E.~Telatar, and E.~{Ar\i kan}, ``Polarization for
  arbitrary discrete memoryless channels,'' in \emph{Proc. of the IEEE
  Information Theory Workshop}, Taormina, Sicily, october 2009, pp. 144--148.

\bibitem{MT10}
R.~Mori and T.~Tanaka, ``Channel polarization on $q$-ary discrete memoryless
  channels by arbitrary kernel,'' in \emph{Proc. of the IEEE International
  Symposium on Information Theory}, Austin, Texas, June 2010, pp. 894--898.

\bibitem{PB13}
W.~Park and A.~Barg, ``Polar codes for $q$-ary channels, $q= 2^r$,'' \emph{IEEE
  Trans. on Information Theory}, vol.~59, no.~2, pp. 955--969, February 2013.

\bibitem{SaP13}
A.~G. Sahebi and S.~S. Pradhan, ``Multilevel channel polarization for arbitrary
  discrete memoryless channels,'' \emph{IEEE Trans. on Information Theory},
  vol.~59, no.~12, pp. 7839--7857, December 2013.

\bibitem{GohariA12}
A.~A. Gohari and V.~Anantharam, ``Evaluation of {Marton's} inner bound for the
  general broadcast channel,'' \emph{IEEE Trans. on Information Theory},
  vol.~58, no.~2, pp. 608--619, February 2012.

\bibitem{GengJNW13}
Y.~Geng, V.~Jog, C.~Nair, and Z.~V. Wang, ``An information inequality and
  evaluation of {Marton's} inner bound for binary-input broadcast channels,''
  \emph{IEEE Trans. on Information Theory}, vol.~59, no.~7, pp. 4095--4105,
  July 2013.

\bibitem{GohariNA_isit2013}
A.~Gohari, C.~Nair, and V.~Anantharam, ``Improved cardinality bounds on the
  auxiliary random variables in {Marton's} inner bound,'' in \emph{Proc. of the
  IEEE International Symposium on Information Theory}, Istanbul, Turkey, July
  2013, pp. 1272--1276.

\bibitem{GengGNY14}
Y.~Geng, A.~Gohari, C.~Nair, and Y.~Yu, ``On {Marton's} inner bound and its
  optimality for classes of product broadcast channels,'' \emph{IEEE Trans. on
  Information Theory}, vol.~60, no.~1, pp. 22--41, January 2014.

\bibitem{GengN14}
Y.~Geng and C.~Nair, ``The capacity region of the two-receiver {G}aussian
  vector broadcast channel with private and common messages,'' \emph{IEEE
  Trans. on Information Theory}, vol.~60, no.~4, pp. 2087--2104, April 2014.

\bibitem{S78}
M.~Salehi, ``{Cardinality bounds on auxiliary variables in multiple-user theory
  via the method of Ahlswede and Korner},'' Stanford {U}niversity, Tech.
  Rep.~33, 1978.

\bibitem{GNSW14}
Y.~Geng, C.~Nair, S.~Shamai, and Z.~V. Wang, ``On broadcast channels with
  binary inputs and symmetric outputs,'' \emph{IEEE Trans. on Information
  Theory}, vol.~59, no.~11, pp. 6980--6989, November 2013.

\bibitem{HKU09a}
N.~Hussami, S.~B. Korada, and R.~Urbanke, ``Performance of polar codes for
  channel and source coding,'' in \emph{Proc. of the IEEE International
  Symposium on Information Theory}, July 2009, pp. 1488--1492.

\bibitem{MajR91}
E.~E. Majani and H.~Rumsey, ``Two results on binary-input discrete memoryless
  channels,'' in \emph{Proc. of the IEEE International Symposium on Information
  Theory}, Budapest, Hungary, June 1991, p. 104.

\bibitem{ShuF04}
N.~Shulman and M.~Feder, ``The uniform distribution as a uniform prior,''
  \emph{IEEE Trans. on Information Theory}, vol.~50, no.~6, pp. 1356 -- 1362,
  Jun. 2004.

\bibitem{Gal68}
R.~G. Gallager, \emph{Information Theory and Reliable Communication}.\hskip 1em
  plus 0.5em minus 0.4em\relax New York: Wiley, 1968.

\bibitem{SRDR12}
D.~Sutter, J.~M. Renes, F.~Dupuis, and R.~Renner, ``Achieving the capacity of
  any {DMC} using only polar codes,'' in \emph{Proc. of the IEEE Information
  Theory Workshop}, Lausanne, Switzerland, September 2012, pp. 114--118.

\bibitem{HY13}
J.~Honda and H.~Yamamoto, ``Polar coding without alphabet extension for
  asymmetric models,'' \emph{IEEE Trans. on Information Theory}, vol.~59,
  no.~12, pp. 7829--7838, December 2013.

\end{thebibliography}


\end{document}